\documentclass[conference]{IEEEtran}
\IEEEoverridecommandlockouts
\usepackage{cite}
\usepackage{amsmath,amssymb,amsfonts}
\usepackage{algorithmic}
\usepackage{graphicx}
\usepackage{textcomp}
\usepackage{xcolor}
\usepackage{amsthm}
\def\BibTeX{{\rm B\kern-.05em{\sc i\kern-.025em b}\kern-.08em
    T\kern-.1667em\lower.7ex\hbox{E}\kern-.125emX}}

\usepackage[linesnumbered ,titlenumbered,ruled,noend]{algorithm2e}
\SetKwRepeat{Do}{do}{while}

\newcommand{\ouralgorithm}{FunShare}

\newtheorem{challenge}{Challenge}
\newtheorem{definition}{Definition}
\newtheorem{problem}{Problem}
\newtheorem{theorem}{Theorem}

\usepackage{balance}

\usepackage{subcaption}

\usepackage{url}

\usepackage[skins]{tcolorbox}
\newtcolorbox{takeaway}[2][]{%
  enhanced,
  colback=white,
  colframe=black,
  coltitle=black,
  sharp corners,
  boxrule=0.4pt,
  left=4pt, right=4pt,top=4pt,bottom=4pt, 
  enlarge top by=-0.5\baselineskip, enlarge bottom by=\baselineskip, 
  attach boxed title to top left={yshift=-0.5\baselineskip,xshift=4pt},
  boxed title style={tile,size=minimal,left=2pt,right=2pt,
    colback=white,before upper=\strut},
  title=#2,#1
}

\begin{document}

\title{Process Faster, Pay Less:\\Functional Isolation for Stream Processing}

\author{
  \IEEEauthorblockN{%
    Eleni Zapridou\IEEEauthorrefmark{1},
    Michael Koepf\IEEEauthorrefmark{2}\thanks{The work was done while Koepf, Sioulas, and Mytilinis were at EPFL.},
    Panagiotis Sioulas\IEEEauthorrefmark{3},
    Ioannis Mytilinis\IEEEauthorrefmark{3},
    Anastasia Ailamaki\IEEEauthorrefmark{1}}
  \IEEEauthorblockA{\IEEEauthorrefmark{1}EPFL, Switzerland
  \IEEEauthorrefmark{2}TU Wien, Austria
  \IEEEauthorrefmark{3}Oracle, Switzerland}
  \IEEEauthorblockA{
  \{name.surname\}@epfl.ch,
michael.koepf@student.tuwien.ac.at,
\{name.surname\}@oracle.com}
  
}

\maketitle

\begin{abstract} 
Concurrent workloads often extract insights from high-throughput, real-time data streams. Existing stream processing engines isolate each query's resources, ensuring robust performance but incurring high infrastructure costs. In contrast, sharing work reduces the amount of necessary resources but introduces inter-query interference, leading to performance degradation for some queries.

We introduce FunShare, a stream-processing system that improves resource efficiency without compromising performance by dynamically grouping queries based on their performance characteristics. 
FunShare strategically relaxes query interdependencies and minimizes redundant computation while preserving individual query performance. It achieves this by using an adaptive optimization framework that monitors execution metrics, accurately estimates computation overlaps, and reconfigures execution plans on the fly in response to changes in the underlying data streams. Our evaluation demonstrates that FunShare minimizes resource consumption compared to isolated execution while maintaining or improving throughput for all queries.

\end{abstract}


\setlength{\textfloatsep}{8pt plus 0pt minus 2pt}

\vspace{-2px}
\section{Introduction}

Organizations increasingly rely on real-time data streams to make rapid decisions in applications ranging from user-experience customization to service monitoring and autonomous driving. Streaming applications have stringent performance requirements while processing data with highly variable distributions and rates. Maintaining high Quality of Service (QoS) requires that streaming queries maintain continuous access to the necessary compute resources. However, as applications grow in complexity and demand deeper insights from data, the increasing volume of queries creates a challenging trade-off between resource consumption and QoS. 

At one end of the design spectrum, existing Stream Processing Engines (SPEs) adopt resource isolation~\cite{storm-parallelism, storm-scheduler, hazelcast-engine, kafka-architecture}. By allocating dedicated resources to each query, they meet QoS requirements and ensure predictable performance. However, this approach comes at the cost of excessive resource consumption. Moreover, resource isolation leads to poor scalability, as the infrastructure must grow proportionally with the number of queries, resulting in very high annual costs \cite{DBLP:journals/pvldb/JindalKRP18}. 

On the other end of the spectrum, work sharing has proven to be an effective way to increase scalability and reduce resource usage in highly concurrent workloads~\cite{roulette, datapath, qpipe}.
For instance, in Microsoft's data centers, 45\% of the jobs share common computations that, if shared, could reduce machine hours by up to 40\%~\cite{DBLP:journals/pvldb/JindalKRP18, DBLP:journals/pvldb/MaiZPXSVCKMKDR18}. 
In work sharing, multiple queries are processed using a single plan that executes common operators only once. While this reduces overall resource demand, it comes at the cost of QoS. Since queries are processed with a common plan, their performance becomes interdependent, which can adversely affect some queries---especially those with strict requirements.
For example, consider two filter--join queries \textit{Q1, Q2} that share the same join but differ in selectivity: \textit{Q1} is highly selective, whereas \textit{Q2} processes all the tuples. In shared execution, both queries observe the same throughput, potentially violating the QoS requirements of the highly selective \textit{Q1}.

While previous research has explored work-sharing techniques for SPEs, it focuses on optimizing aggregate data processing throughput and overlooks the effect of sharing on the QoS of individual queries~\cite{chandrasekaran_telegraphcq_2003, astream, ajoin}. As a result, these techniques are unsuitable for streaming applications with strict QoS requirements.
To enable resource sharing while mitigating the adverse effects of query interdependencies, we adopt the principle of functional isolation~\cite{cidr-oligolithic}: “systems should cross-optimize queries only if individual performance is not compromised.” Prior work\cite{cidr-oligolithic} shows that functional isolation can be achieved in batch analytics by carefully forming sharing groups of queries, reducing CPU utilization without degrading the performance of any individual query.

However, the partitioning-once approach of queries into sharing groups~\cite{cidr-oligolithic} is insufficient for streaming environments due to the dynamic, unpredictable nature of real-time data streams. Both the input rate and distribution of the underlying stream can be unpredictable, causing the efficiency of a particular group's plan---as well as the performance of each query in isolation---to fluctuate over time. Ensuring functional isolation in such an environment requires: i) accurately capturing the trends of the streams and the potential benefits/risks of grouping, ii) continuously adjusting grouping decisions at runtime, and iii) reconfiguring the system with minimal cost. By doing so, we can achieve both resource efficiency and improved QoS simultaneously.

We propose \ouralgorithm{}, a stream processing system that uses functional isolation in order to minimize resource consumption and, at the same time, guarantee high QoS for individual queries. \ouralgorithm{} organizes queries into \emph{dynamic sharing groups}---sets of queries that can safely share work without compromising performance isolation at a given moment---and identifies the minimum amount of compute resources to assign to each group such that each query's QoS requirements are met. \ouralgorithm{} employs a novel mechanism that leverages runtime information to accurately capture i) data shifts and ii) the benefit/risk of grouping. Hence, it can incrementally and safely refine sharing groups, optimizing resource utilization while giving the impression of isolated execution. 

This paper makes the following contributions:
\begin{itemize}
    \item We define the functional isolation principle for streams; Resource consumption is minimized while preserving high QoS for individual queries. This way, we show that sometimes there is free lunch and stream processing becomes more scalable and affordable at no cost.

    \item We propose an adaptive mechanism that continuously (re-)partitions streaming queries into sharing groups, enabling them to execute with significantly fewer resources without performance penalties compared to isolated execution. Our experiments show that \ouralgorithm{} uses $1-10.7\times$ fewer resources than isolated execution while achieving the same or higher throughput for all queries.

    \item The key enabler of our solution is a novel performance estimator that leverages runtime information to accurately predict the performance of hypothetical group formations. In contrast to existing solutions, the estimator scales to the number of groups and is robust to changes in the underlying stream. 
\end{itemize}

\section{The Triangle of Concurrency, Quality-of-Service and Resources}

Business functions or applications are often time-sensitive, requiring rapid insights from real-time events. At the same time, not all queries are of the same importance, nor do they exhibit the same performance characteristics. For example, in a self-driving vehicle manufacturer scenario, we can consider two types of queries: (i) \emph{Alerts}, which filter a single stream of sensor data, are lightweight to compute but safety-critical, and (ii) \emph{A/B testing} for evaluating the driving software/hardware, which joins multiple streams, is heavier to compute but enjoys relaxed responsiveness constraints. As multiple diverse operations may be executed \textbf{concurrently} on a \textbf{shared} infrastructure, interference effects can violate individual query requirements. To quantify this effect, we use the notion of \emph{Quality-of-Service (QoS)}.

\begin{definition}[Quality of Service]
The degree to which individual query performance requirements are satisfied.
\end{definition}

In our example, consider the case where an alert for hitting a pedestrian is not produced on time because resources are busy with A/B testing. This is a low QoS scenario.

In this section, we present the interplay among (i) concurrency, (ii) QoS, and (iii) resources, and explain why existing approaches fall short in dynamic streaming environments.

\subsection{Resource Isolation: The Expensive Standard}
\label{ssec:2_1}

A streaming query corresponds to a dataflow---usually in the form of a directed acyclic graph---where nodes represent operators and edges represent data streams. Each operator has an input and an output queue. If the operator cannot sustain the input rate, it applies backpressure, and tuples accumulate in the output queue of the upstream operator. Backpressure can propagate up to the source; there, tuples are typically stored and processed with a delay. If this effect persists for a prolonged period, latency grows unboundedly.

To cope with such situations and sustain the data velocity, distributed stream processing systems employ a data-parallel approach. Each operator has multiple instances, called \emph{subtasks}, and each subtask processes a partition of the stream. Increasing the number of subtasks results in higher parallelism and higher processing throughput.

Different SPEs use slightly different mechanisms to allocate memory and cores for each subtask. For example, Apache Flink allows users to specify the number of task slots per worker, with each slot receiving an equal share of the allocated memory~\cite{flink-architecture}. Each task slot gets assigned at most one subtask from each operator. Kafka Streams uses stream tasks, the number of which per worker is configurable, as the units of its parallelism model~\cite{kafka-architecture}.

In this work, we focus on compute resources. To abstract away implementation specifics, for the remainder of the paper, we use the following definition of \emph{resources}:

\begin{definition}[Resources]
\label{def:resources}
The total number of subtasks for computing the dataflow of a streaming query.
\end{definition}

The amount of required resources is determined a priori on a per-query basis, in a manner that satisfies QoS requirements while minimizing overprovisioning as much as possible. To deal with unpredictable fluctuations both in the input rate and data distribution, SPEs rely on autoscaling features \cite{flink-autoscaling-reactive-mode, flink-k8s-operator-autoscaler, google-dataflow-autoscaling, hazelcast-autoscaling} to adjust resources to the demand and adaptive partitioning to dynamically balance load and mitigate skew~\cite{dataflow-work-rebalancing, dalton}.

However, in reality, streaming infrastructures do not execute a single query but thousands of concurrent queries~\cite{DBLP:journals/pvldb/MaiZPXSVCKMKDR18, uber-real-time-infra}. Existing SPEs favor QoS and enforce resource isolation. Maintaining each query's QoS requires provisioning sufficient infrastructure to meet the aggregate resource reservations. Hence, resource allocation is managed separately for each query, with queries executing in isolation using dedicated resources. Then, concurrency directly translates into increased monetary and maintenance costs. \emph{Therefore, existing SPEs maintain high-standard QoS, by increasing resources and monetary costs proportionally to query concurrency.}

\subsection{Taming Concurrency using Work Sharing}

\begin{figure}[t]
 \centering
    \includegraphics[width=1\linewidth]{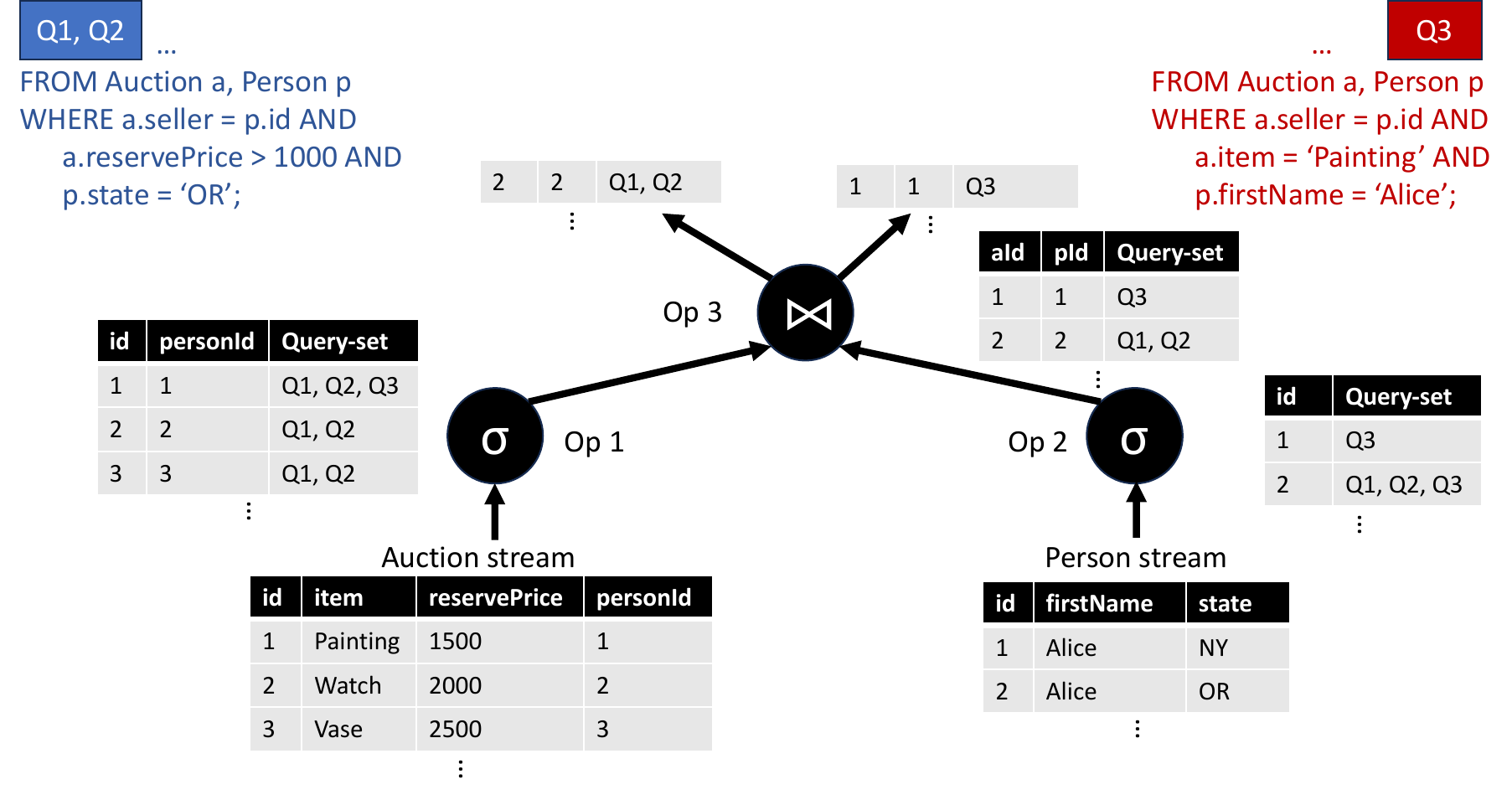}
    \caption{Sharing using a global plan and the Data--Query model}
	\label{fig:data_query_model}
\end{figure}

An alternative to accumulating resource allocations is to reduce redundant processing across queries and, thus, the resource requirements using work sharing. Work-sharing systems process multiple queries at once (Fig.~\ref{fig:data_query_model}) by leveraging (i) a \emph{global plan} and (ii) the \emph{Data--Query model \cite{giannikis_shareddb_2012}}. 

\textbf{Global plan:} The global plan expresses sharing opportunities among different queries. It is a directed acyclic graph (DAG) of operators that process tuples for one or more queries, and multicast their results to one or more downstream operators. 
For example, Fig.~\ref{fig:data_query_model} shows part of the global plan for queries Q1, Q2, and Q3. Shared operators evaluate filters (op. 1 and 2) and execute the join (op. 3) for all queries at once. The join result is routed to different downstream operators for Q1 and Q2, and Q3, respectively. Prior work shows that global plans can share operators such as scan, filter, join, sort, top-N, and group-by \cite{giannikis_shareddb_2012}, as well as window operators \cite{on-the-fly-sharing}.

By processing each operator of the global plan only once, the system shares work across queries, reducing both overall computation and required resources.

\textbf{Data--Query model:} The Data--Query model is a key technique for sharing operators between queries with different selection predicates \cite{giannikis_shareddb_2012, astream, ajoin, roulette}. It annotates each tuple with a \textit{query set} indicating which queries the tuple contributes to. \emph{Shared} operators in the global plan process both the actual tuples and the query sets. The query set tracks membership for intermediate results, controls the routing of tuples, and helps eliminate redundant tuples early.

Fig.~\ref{fig:data_query_model} shows the query sets for intermediate results. In operator 1, tuples are tagged with query sets based on the predicates. The auction with id 1 satisfies all predicates and has all three queries in the query set, whereas the other two fail the predicate \textit{item = "Painting"} and exclude Q3. In the join, query sets are cross-checked; hence, the tuple with auction id = 1 retains only Q3 in its query set, and the tuple with id = 3 is eliminated. Finally, join results are routed to downstream operators based on their query sets. Using the Data-Query model: (i) the global plan can share work on tuples common across some, but not all, queries, and (ii) operators can immediately drop tuples that do not belong to any query.

\subsection{The Challenges of QoS-Friendly Sharing}

Although work sharing reduces overall resource demand, it does not provide guarantees about individual query performance: some queries will experience speedups, but others---particularly lightweight queries that would run quickly in isolation---may be slowed down. In the example from Section \ref{ssec:2_1}, this translates to worse QoS for the critical alerts.

Consider a workload with queries that share exactly the same plans, except for the filters: there are highly selective and non-selective queries. Each query is configured with an amount of resources adequate to sustain the input rate. On the one hand, isolated execution maintains a constant average throughput at the cost of increased resource usage. On the other hand, work sharing will reduce resource consumption, but average throughput can drop significantly, indicating QoS violation for specific queries.

To bridge this performance gap, prior work suggests to selectively share work based on query access patterns: execution should be shared only among queries that process ``many'' common tuples. Selectivity-based approaches (e.g., SWO~\cite{swo}) solve a classification problem: data accesses are classified into high (\textit{H}) and low selectivity (\textit{L}). Then, sharing decisions depend on the class of each input and the selected algorithm for computing the joins. A major drawback of this approach is that it does not account for inter-query correlations. For our running example, consider a batch of $100$ alerting queries that randomly access $1\%$ of the data. SWO will classify them as \textit{L} and share execution. However, each query is expected to process a different $1\%$ region of the data, so the shared operator eventually processes the entire stream.

To remedy this, AJoin~\cite{ajoin} tries to capture the actual work overlap for a set of queries. 
However, the analytical formula used to calculate query overlaps does not scale with the number of groups. 
Additionally, AJoin makes sharing decisions solely based on minimizing total computational cost; it shares execution if running two queries together is cheaper than running them separately. 
This approach can lead to suboptimal sharing choices, such as pairing a highly selective query with a low-selectivity one, provided they have significant filter overlap. 
This may degrade the performance of the more selective query due to increased processing overhead.

\begin{figure}[t]
 \centering
    \includegraphics[width=0.9\linewidth]{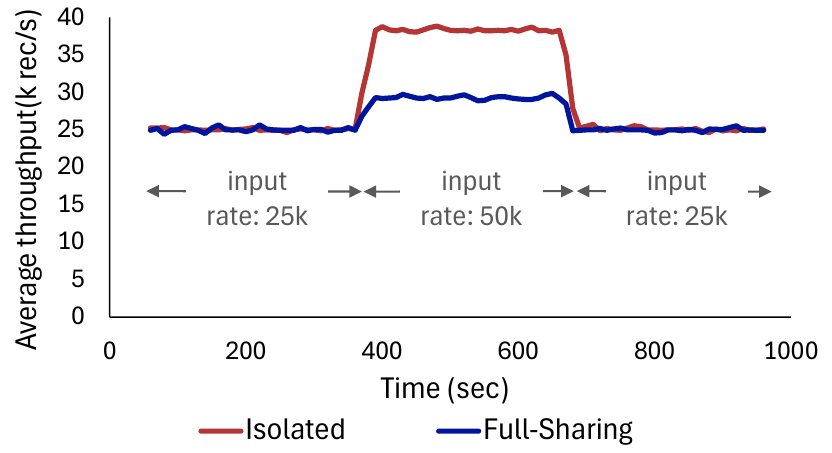}
    \caption{Average throughput during input rate shifts. The workload includes queries with different downstream operations: GROUP BY average and a heavy-weight UDF}
    \label{fig:sharing_downstream_issue}
\end{figure}

\begin{challenge}[Access pattern vs Resources]
Smart sharing decisions that do not compromise QoS depend on data--query--resources correlations, which are hard to capture, especially in a streaming environment where data is unpredictable and distribution may shift at any time.
\end{challenge}

In practice, queries differ in more than just selectivity, and data distributions and input rates fluctuate over time.
When deciding whether to share an operator, we should also consider the downstream operators of each query. For example, consider two A/B testing queries, \textit{Q1} and \textit{Q2} from our self-driving example, that are highly selective and access the same set of tuples. While they share a join, after the join, \textit{Q1} computes a simple \textit{GROUP BY} average, while \textit{Q2} computes a complex, compute-bound UDF. The UDF can apply backpressure on the shared join and impact the throughput and QoS of \textit{Q1}. 

Fig. \ref{fig:sharing_downstream_issue} presents such an example. We use 128 concurrent queries, and resources are carefully assigned to sustain the input rate (this experiment uses the same setup as Section~\ref{sec:eval}:Fig.~\ref{fig:exp:adaptivityIR}). Then, at time 360 seconds, a burst in the input rate in the form of a pulse is observed. As the complex UDF cannot sustain the new rate, in isolated execution, the throughput relative to the input rate is expected to drop by $\frac{\textit{num\_heavy\_queries}}{\textit{num\_total\_queries}}  \times (1 - \frac{\textit{UDF\_throughput}}{\textit{input\_rate}})$. However, in shared execution, lightweight queries will also be affected, so it will drop by $1 - \frac{\textit{UDF\_throughput}}{\textit{input\_rate}}$. Similar effects would have been observed in the case of a distribution shift that caused the downstream operator to become skewed, degrading its processing throughput.

\begin{challenge}[Query Plans vs Data Fluctuations vs Resources]
Sharing decisions should consider global plan complexity and resource requirements rather than operator-local criteria. Moreover, they must be continuous and adaptive, reacting quickly to shifts in input rates or distributions.
\end{challenge}

All in all, workload concurrency in a streaming system exposes different QoS--resources trade-offs that are hard to optimize. We argue that work sharing is the key solution for exploring these trade-offs. However, sharing has its own pitfalls, and prior work has fallen short in efficiently overcoming them. By tackling the challenges posed by the interplay among data--query correlations, dynamic data fluctuations, and resource demand, we advance the work-sharing paradigm for stream processing and deliver high QoS at low cost.

\section{Resource-Efficient Execution in Islands of Queries}
\label{sec:formulation_overview}

\subsection{Functional Isolation for Streams}

To address the challenges that cross-query optimization introduces with respect to QoS, the principle of \emph{functional isolation}~\cite{cidr-oligolithic} has been proposed: any cross-query optimization (e.g., work sharing) should never compromise individual performance. GroupShare~\cite{cidr-oligolithic} implements functional isolation for batch analytics, using runtime performance metrics and an iterative top-down partitioning algorithm. It initially places all queries in a single group and continuously restructures them into smaller groups until it converges to a grouping that respects individual QoS requirements.

However, stream processing marks a significant departure from warehousing data and batch analytics. Both data distribution (hence, query overlaps) and velocity may present high variability. This complicates sharing decisions. A distribution shift may increase the correlation between two queries and increase the sharing benefit, or decrease it and make sharing detrimental. Similarly, changes in velocity may require reconsidering grouping decisions. GroupShare's one-way, divisive partitioning fails in this context; it converges to a static state, and thus queries remain grouped only if sharing is beneficial throughout the entire execution.
Clearly, with ever-changing data trends, grouping decisions need to be continuously reassessed.

We propose \ouralgorithm{}, which advances the work-sharing paradigm by achieving functional isolation for queries on streaming data with varying distributions and velocities. Unlike one-way descent, \ouralgorithm{} continuously re-evaluates prior grouping decisions, rapidly adapting group formations in response to distribution and velocity shifts throughout the stream's lifetime.

\ouralgorithm{} supports queries with arbitrary operators, but sharing candidates are only joins with varying selection predicates. We restrict to joins as they are highly sensitive to query overlaps and provide a rigorous testbed for grouping logic.

We concretely define \ouralgorithm{}'s goal: 

\textbf{Goal:} \textit{Minimize a target workload's resource usage while providing equal or better QoS for each query.}

Let us formally define the problem. Consider queries $q_1, q_2, \dots, q_n$, which run in isolation, consume a stream of rate $D(t)$, and each query's $q_i$ processing throughput is $T_i(t)$. In this work, we assume that QoS is expressed in terms of throughput. To achieve $T_i(t)$, each query $q_i$ requires resources $Resources(q_i)$. Both the input rate and the throughput are functions of time. For resources, we assume careful, a priori provisioning for each query. If $T_i(t) \geq D(t)$, a query can sustain the input rate, but if $T_i(t) < D(t)$, backpressure will be applied and tuples will be accumulated at the source. Then, functional isolation can be re-defined for stream processing as:

\begin{definition}[Functional Isolation for Streams]
\label{def:func_isolation}
Continuously (re-)partition a set of queries \( \{q_1, \dots, q_n\} \) into non-overlapping groups \( g_1, \dots, g_m \) (with
$\bigcup_{j=1}^{m} g_j = \{q_1, \dots, q_n\}$) and execute the shared plan for each group \( g_j \) using a shared resource allocation \( Resources(g_j) \) that achieves a processing rate \( T_{g_j}(t) \) satisfying $T_{g_j}(t) \geq T_{q_{j,k}}(t), \quad \forall\, q_{j,k} \in g_j$.
\end{definition}

Functional isolation formally defines a problem and the framework for the solution, but does not specify any requirements with respect to resource consumption. As our goal is to minimize resources, we pose and solve the following problem:

\begin{problem}[Resource-Efficient QoS Guarantees]
\label{prob:funcisol}
Continuously (re-)partition a set of queries \( \{q_1, \dots, q_n\} \) into non-overlapping groups \( g_1, \dots, g_m \) (with \(\bigcup_{j=1}^{m} g_j = \{q_1, \dots, q_n\}\)) and assign to each group a resource allocation \( Resources(g_j) \) such that the total amount of resources is minimized:
\[
\min \quad \sum_{j=1}^{m} \text{Resources}(g_j)
\]
subject to the constraints:

\noindent(1) For each group \( g_j \), the effective processing rate satisfies
    \[
    T_{g_j}(t) \geq T_{q_{j,k}}(t), \quad \forall\, q_{j,k} \in g_j.
    \]
    
\noindent(2) The shared resource allocation does not exceed the sum of the isolated resource requirements:
    \[
    Resources(g_j) \leq \sum_{q_{j,k} \in g_j} Resources(q_{j,k}), \forall\, j \in \{1, \dots, m\}.
    \]
\end{problem}

The first constraint enforces functional isolation, while the second establishes an upper bound on resource allocation. Shared execution must not consume more resources than isolated execution, as this contradicts the goals of sharing. Furthermore, while not explicitly formalized, an implicit requirement is that any runtime reconfiguration should incur minimal overhead and not affect QoS. By solving Problem \ref{prob:funcisol}, we make stream processing more resource-efficient while providing the same individual QoS guarantees.

\subsection{System Overview}

\begin{figure}
    \centering
    \includegraphics[width=1\linewidth]{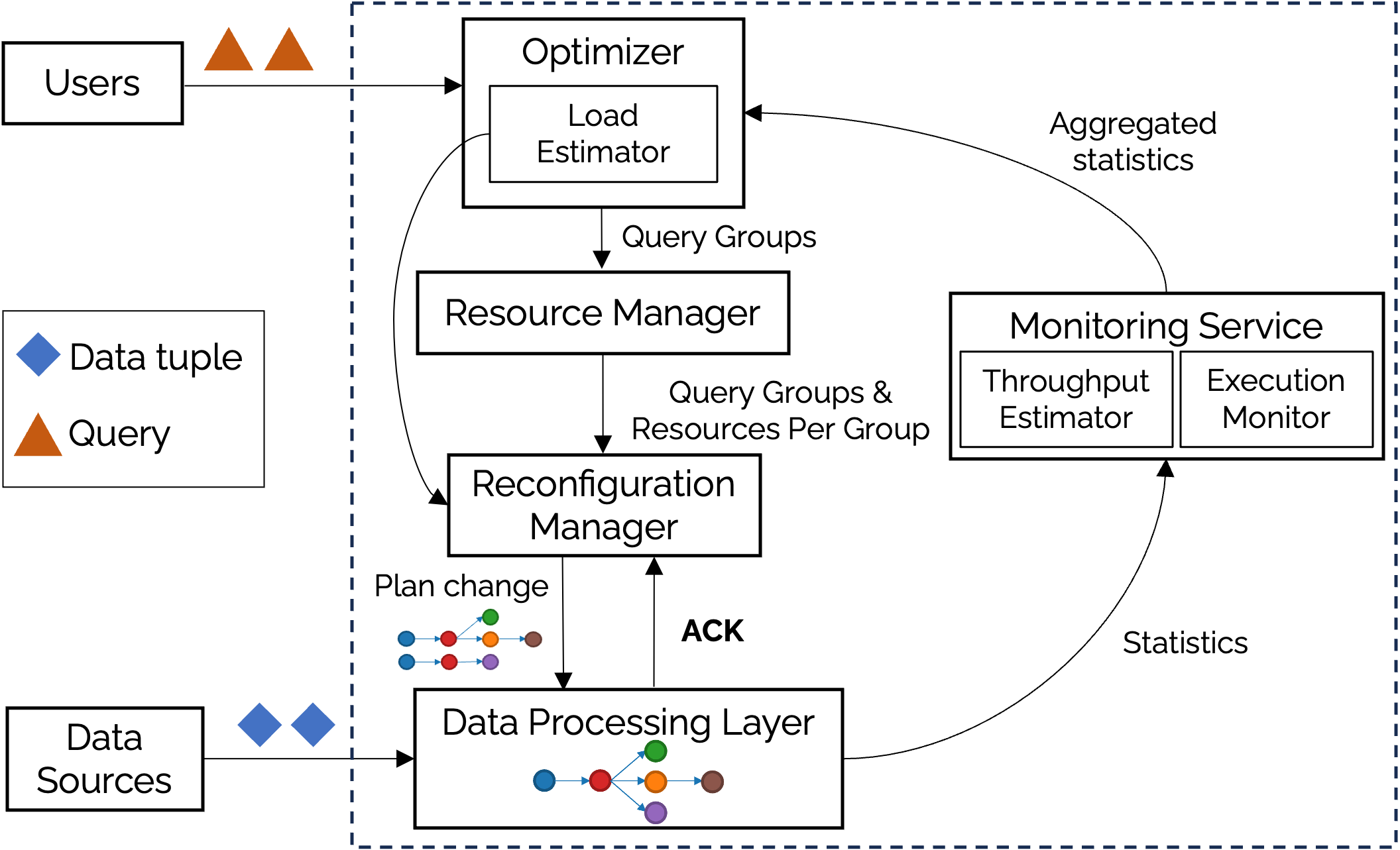}
    \caption{\ouralgorithm{} architecture}
    \label{fig:architecture}
\end{figure}

Fig.~\ref{fig:architecture} illustrates the architecture of \ouralgorithm{}, which incorporates real-time load estimation and monitoring to continuously adapt its execution strategy to dynamic workloads, optimizing resource utilization while ensuring individual query performance remains stable or improves.
The system operates through a continuous feedback loop that monitors query execution, assesses query overlaps, and refines group assignments in real time.
The \emph{Optimizer} receives queries with their resource specifications and analyzes runtime statistics to form and update their partitioning into sharing groups, ensuring functional isolation. 
Whenever the sharing groups change, the \emph{Optimizer} sends them to the \emph{Resource Manager}, the component responsible for deciding each group's resource assignment.
The \emph{Reconfiguration Manager} orchestrates plan adjustments on the fly without pausing data processing.
The data processing layer executes each group's global plan, employing the data--query model to share computation between queries with varying selection predicates. 
The \emph{Monitoring Service} collects statistics about the execution and workload at runtime and aggregates them, providing the \emph{Optimizer} with up-to-date insights to inform sharing decisions. Finally, the \emph{Optimizer} has a \emph{Load Estimator} that employs a novel mechanism to estimate at runtime the computation load of groups and the overlap between them.

\section{Adaptive Sharing-Group Partitioning}
\label{sec:optimizer}

To achieve the objective of Problem~\ref{prob:funcisol}---reduce resource consumption while ensuring that individual query performance remains uncompromised---\ouralgorithm{} dynamically adjusts group assignments at runtime using two mechanisms. First, group merging periodically evaluates whether, under the current data distribution, combining two groups can improve performance or reduce resource usage without diminishing any query's throughput. 
Second, group splitting is triggered when a query experiences performance degradation within its group due to a workload shift. Then, smaller groups are created to promptly restore performance isolation.
Overall, splitting and merging complement each other to dynamically maintain an effective grouping: splitting quickly isolates queries that perform worse within a group than in isolation, while merging updates groups when sharing does not incur penalties.

In this section, we first present the group merging and splitting mechanisms (Sections~\ref{sec:merge} and~\ref{sec:split}). Then, we describe the function of the \emph{Resource Manager}, which decides on the per-group resource allocations (Section~\ref{sec:resource-allocation}). Finally, we present the mechanisms for estimating common load between groups and query performance in isolation, and for collecting execution metrics (Section~\ref{sec:monitoring}).

\subsection{Group Merging Mechanism}
\label{sec:merge}
The merging mechanism periodically combines groups to eliminate redundant processing and reduce overall resource consumption. 
Merging decisions must balance the benefits of shared computation against the risk of degrading individual query throughput.

The key intuition behind merging two groups is to ensure that the additional computation load introduced by the merge can be handled without penalizing either group. Specifically, a group $g_j$ should only be merged with a group $g_i$ if the extra work imposed by the merge (computation $g_j$ does not need) can be offset by available resources. Such resources originate from two sources: (i) resources currently allocated to $g_i$, and (ii) resources allocated to $g_j$ that are idle. This is consistent with the second constraint of Problem~\ref{prob:funcisol}. Ensuring this condition for both groups is critical---otherwise, the merge could degrade throughput for one of the queries.

To quantify the impact of a merge decision, we define the cost metric $GroupingCost(g_i, g_j)$ that captures the additional processing load imposed on group $g_j$ by merging it with group $g_i$, relative to the available resources to process it. 

\begin{equation} 
\label{eq:grouping_cost} 
\mathit{GroupingCost}(g_i, g_j) = \frac{ \frac{\mathrm{Load}(g_i \cup g_j) - \mathrm{Load}(g_j)}{\mathrm{Load}(g_i \cup g_j)}} { \frac{Resources(g_i) + \mathit{IdleResources}(g_j)}{Resources(g_i) + Resources(g_j)}} 
\end{equation}

\noindent where $Load(g)$ represents the processing load of group $g$, 
and $Load(g_i \cup g_j)$ the total load of the merged group $g_i \cup g_j$. $Resources(g)$ expresses the resources allocated to group $g$, and $IdleResources(g)$ the resources that remain unused in $g$. 

The numerator captures the increase in load for group $g_j$ when merging, relative to the combined load of the merged group. 
The denominator measures the proportion of total resources available to process the additional load.
As a result, the metric provides a normalized measure of the additional processing burden relative to the overall resource capacity. 
A lower $GroupingCost(g_i,g_j)$ indicates that the overhead of merging will likely be offset by the available resources, while a higher cost suggests that the merge could decrease the throughput of some queries in group $g_j$. 
Crucially, all components used to compute $GroupingCost$, aside from the resources allocated to each query, depend on the current workload.
Furthermore, accurately estimating $Load(g_i \cup g_j)$ is challenging because it requires relating computations performed by two groups currently executed independently.
We detail how \ouralgorithm{} estimates these quantities accurately yet efficiently in Section~\ref{sec:monitoring}.

\begin{algorithm}
\SetAlgoLined
\SetKwInOut{Input}{Input}
\SetKwInOut{Local}{Local}
\Local{Set of groups \(G\)}
\Local{Boolean \(merging\_possible \gets true\)}

\While{\(merging\_possible\)}{
    \(min\_cost \gets \infty\)\;
    \(groups\_to\_merge \gets None\)\;
    \(merging\_possible \gets false\)\;
    
    \ForEach{pair \((g_i, g_j)\) in \(G\) with a common operator}{
        \If{No downstream backpressure propagation is expected from merging \(g_i\) and \(g_j\)}{
            \label{alg:line:merge-backpressure-check}
            
            \(
            \begin{aligned}
              cost \gets {} & \max\Bigl(\operatorname{GroupingCost}(g_i, g_j),\\
                                & \operatorname{GroupingCost}(g_j, g_i)\Bigr)
            \end{aligned}
            \)\;
            
            \label{alg:line:merge-cost-calc}
            \If{\(cost < min\_cost\) and \(cost < MERGE\_THRESHOLD\)}{ \label{alg:line:merge-threshold}
                \(min\_cost \gets cost\)\;
                \(groups\_to\_merge \gets (g_i, g_j)\)\;
                \(merging\_possible \gets true\)\;
            }
        }
    }
    
    \If{\(groups\_to\_merge \neq None\)}{
        \(G.update\_groups(groups\_to\_merge)\)\;
    }
}
merge\_groups(G)
\caption{Group Merging (Minimizing resources)}
\label{alg:merge}
\end{algorithm}

Algorithm~\ref{alg:merge} outlines the merging procedure.
In each optimization cycle, the system examines all pairs of groups in the current set $G$ that share operations. 
The process begins by assessing the risk of backpressure (line~\ref{alg:line:merge-backpressure-check}). 
Concretely, if the candidate shared operators within the lower-throughput group are already backpressured by their downstream subplan, the merge is skipped to prevent the slow group from throttling the other one. 
Note that checking only pairs of groups may not be optimal, but \ouralgorithm{} does not aim for optimality but rather saving resources while respecting functional isolation and not adding overheads itself.
Next, the algorithm computes the overhead (or cost) that each group would incur by absorbing the other group’s workload (line~\ref{alg:line:merge-cost-calc}). 
As $GroupingCost(g_j,g_i)$ is asymmetric, we estimate the merge cost using $max(GroupingCost(g_j,g_i), GroupingCost(g_i,g_j))$, thereby ensuring that QoS requirements are met for both groups.
A merge is allowed only if this cost is below a preset threshold (line~\ref{alg:line:merge-threshold}). 
Setting the merge threshold to 1 would be a natural choice, indicating that the resource increase should exceed the increase in computation cost. However, the threshold can be set to a lower value to account for performance trends scaling sub-linearly to resource increase. Additionally, lowering the merge threshold results in more conservative merging decisions. 
At every iteration, the merge with the lowest cost is executed, and the process repeats until no further merges are possible. Once the merge step concludes, the \emph{Optimizer} forwards the finalized groups to the \emph{Resource Manager}. Subsequently, the data processing layer will make all necessary query plan changes in a single step.

\subsection{Group Splitting Mechanism}
\label{sec:split}
When shifts in data distribution cause the performance of a sharing group to degrade, group splitting becomes necessary. To maintain functional isolation, the system must rapidly detect these shifts and adjust group assignments accordingly.

The splitting mechanism is triggered when at least one query in a group is predicted to perform worse than if it was running in isolation. (We detail how we detect penalized queries in Section~\ref{sec:resource-allocation}.) The primary objective of splitting is to isolate these queries quickly, preventing QoS degradation. Splitting also serves as a correction mechanism for merging decisions that led to groups that do not respect functional isolation.

\begin{algorithm}
\SetAlgoLined
\SetKwInOut{Input}{Input}
\SetKwInOut{Local}{Local}
\Input{Group $g$, Set of penalized queries $PQ$}
\caption{Group Splitting (Preserving functional isolation)}
\label{alg:split}
\uIf{shared subplan of group $g$ is backpressured} { \label{alg:line:split-backpressure-start}
    split the queries causing backpressure\;
    return\; \label{alg:line:split-backpressure-end}
}
\uElseIf {resource increase for query group possible} {
    request resource increase from resource manager\;
}
\Else{
    place queries in $PQ$ into singleton groups\; \label{alg:line:split-penalized}
}

\end{algorithm}

Algorithm~\ref{alg:split} outlines the procedure, which includes two main steps. In the first step, termed \textbf{backpressure response}, when a shared subplan experiences backpressure, the \emph{Optimizer} identifies the queries causing the slowdown and immediately isolates them (lines~\ref{alg:line:split-backpressure-start}--\ref{alg:line:split-backpressure-end}). Backpressure detection happens based on statistics provided by the \emph{Monitoring Service} (Section~\ref{sec:monitoring}). This ensures that slow queries do not throttle the entire group. In the second step, referred to as \textbf{resource check and isolation}, in the absence of backpressure, the \emph{Optimizer} first checks whether additional resources can be allocated to the group. If additional resources are available, they are requested from the \emph{Resource Manager}. If not, queries expected to be penalized are removed from the group and placed into singleton groups as a quick response that guarantees they meet their QoS. Attempting to add resources before isolating queries preserves the benefits of shared processing while ensuring performance isolation when contention increases. After splitting, the subsequent merging phase may combine singleton groups with others if this does not compromise QoS, thereby enhancing overall resource efficiency. 
We have theoretically proven that \ouralgorithm{} converges (Appendix~\ref{proofs}): assuming no distribution changes, the merge step produces groups that maintain functional isolation. The corrective splitting step is required only when distribution shifts, and even then, functional isolation can be restored in a bounded number of splits.

\subsection{Allocating Resources to Sharing Groups} \label{sec:resource-allocation}

After forming query groups, the \emph{Optimizer} forwards them to the \emph{Resource Manager}, which determines the appropriate resource allocation per group to ensure resource efficiency without penalizing individual queries.

\paragraph{Resource Provisioning during Merging} 
The \emph{Resource Manager} leverages the $GroupingCost$ metric (Equation~\ref{eq:grouping_cost}) to guide its allocation decisions. When merging two or more groups, \(M = \{g_1, \dots, g_m\}\), the manager computes the minimum amount of resources that ensures the cost for all individual groups remains below the merge threshold. Specifically, with $M_{-i} = M \setminus \{g_i\}$, merge threshold as $MT$ and
\[
\begin{aligned}
  &Resources^*(M_{-i}) = \operatorname*{min} \quad Resources(M_{-i}) \\
  &\text{s.t.}\quad  GroupingCost(M_{-i}, g_i; Resources(M_{-i})) < MT
\end{aligned}
\]

\noindent
it solves: $i^* = \operatorname*{argmax}_{i \in \{1,\dots,m\}} Resources^*(M_{-i})$

For the full group, the manager provisions the necessary resources: $Resources(i^*) + Resources^*(M_{-i^*})$

This approach guarantees that no group in the merged set experiences a cost exceeding the selected threshold.

\paragraph{Resource Adjustment upon Query Penalty} 
If a query within a sharing group exhibits throughput degradation relative to its isolated performance, the \emph{Optimizer} can request an increase in the group’s resources.
The provisioning is raised up to the sum of the individual resources that would be allocated if each query was executed in isolation. 
If the group is already operating at its maximum resource allocation, the \emph{Optimizer} proceeds to reassign certain queries into singleton groups.
Consequently, the \emph{Resource Manager} reduces the allocation for the original group so that the split queries can run independently, restoring their performance.

\subsection{Estimating Grouping Cost and Isolated Throughput}
\label{sec:monitoring}

To decide on group merging and splitting, the \emph{Optimizer} relies on two components: the \emph{Load Estimator} and the \emph{Monitoring Service}, which includes the \emph{Throughput Estimator} and the \emph{Execution Monitor} (Fig.~\ref{fig:architecture}).

\paragraph{Load Estimation}
Estimating the computational load for sharing groups accurately is essential to compute the grouping cost (Equation~\ref{eq:grouping_cost}). We need to estimate the individual load of a group, $\mathrm{Load}(g)$, and 
the total load after merging, $\mathrm{Load}(g_i \cup g_j)$.

We introduce a mechanism for efficiently and accurately estimating query and group loads, using sampling, the Data Query Model, and lightweight reconfiguration.
This mechanism enables us to estimate the combined load of groups currently being executed separately, providing valuable insights for evaluating future groupings.
The \emph{Load Estimator} triggers this process before each merging phase.
For each subpipeline that appears in more than one sharing group (a sharing candidate), it selects one group to collect workload statistics. To minimize the overhead of this process, the \emph{Load Estimator} heuristically chooses the group with the highest selectivity.
For a small data sample, the selected group processes all tuples that belong to at least one query, identified through their query set, and calculates workload statistics. To express processing per tuple, \ouralgorithm{} employs cost models. Although we do not advocate for a specific model, our implementation uses an analytical cost model~\cite{DBLP:conf/icde/KangNV03, engine-based-cost-model} as we have found it accurate for the examined workloads, and at the same time, it is easily explainable and cheap to compute.

\begin{figure}
    \centering
    \includegraphics[width=0.8\linewidth]{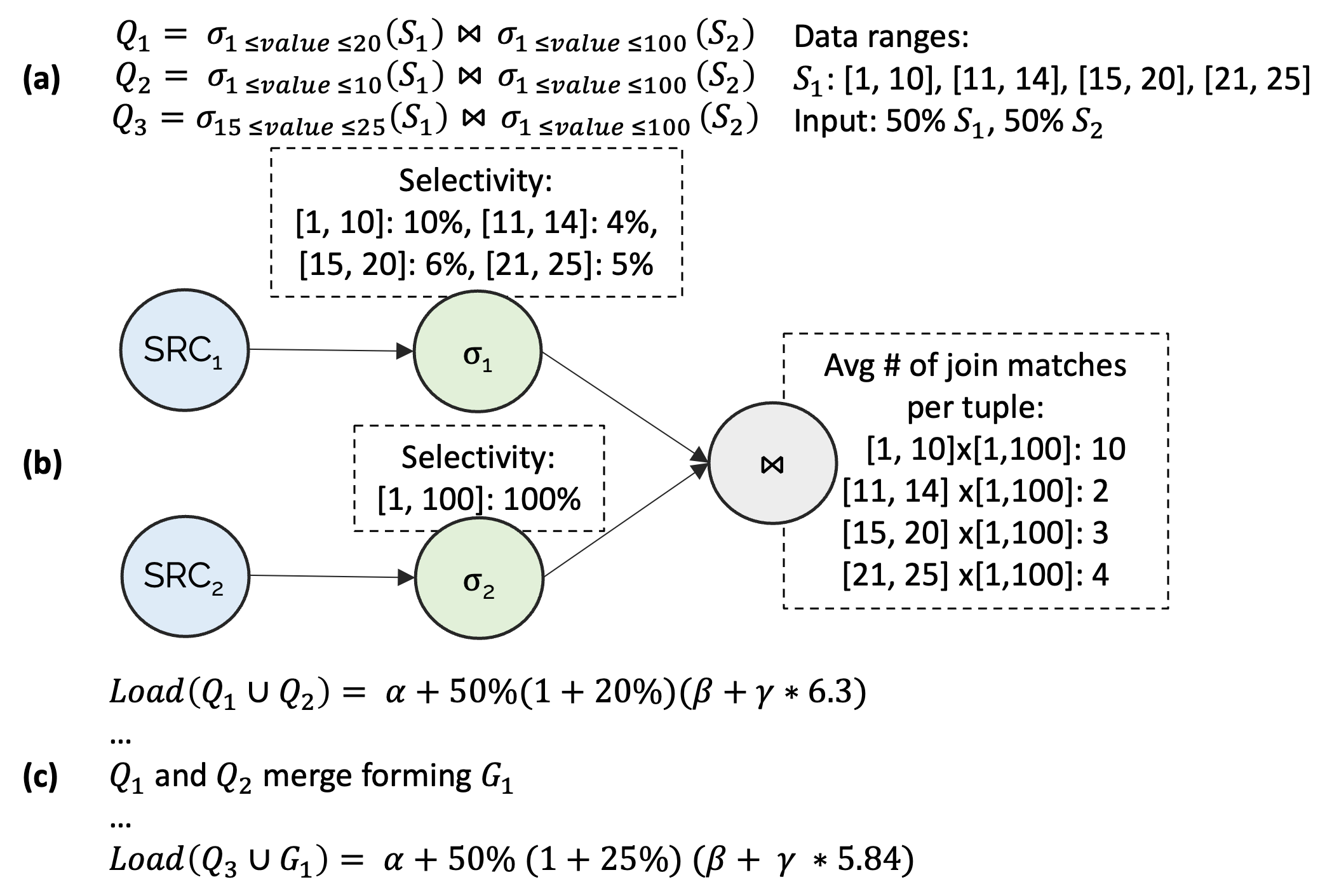}
    \caption{Example of the load estimation mechanism}
    \label{fig:merge-monitoring}
\end{figure}

\begin{figure*} [t]
    \centering
    \includegraphics[width=0.7\linewidth]{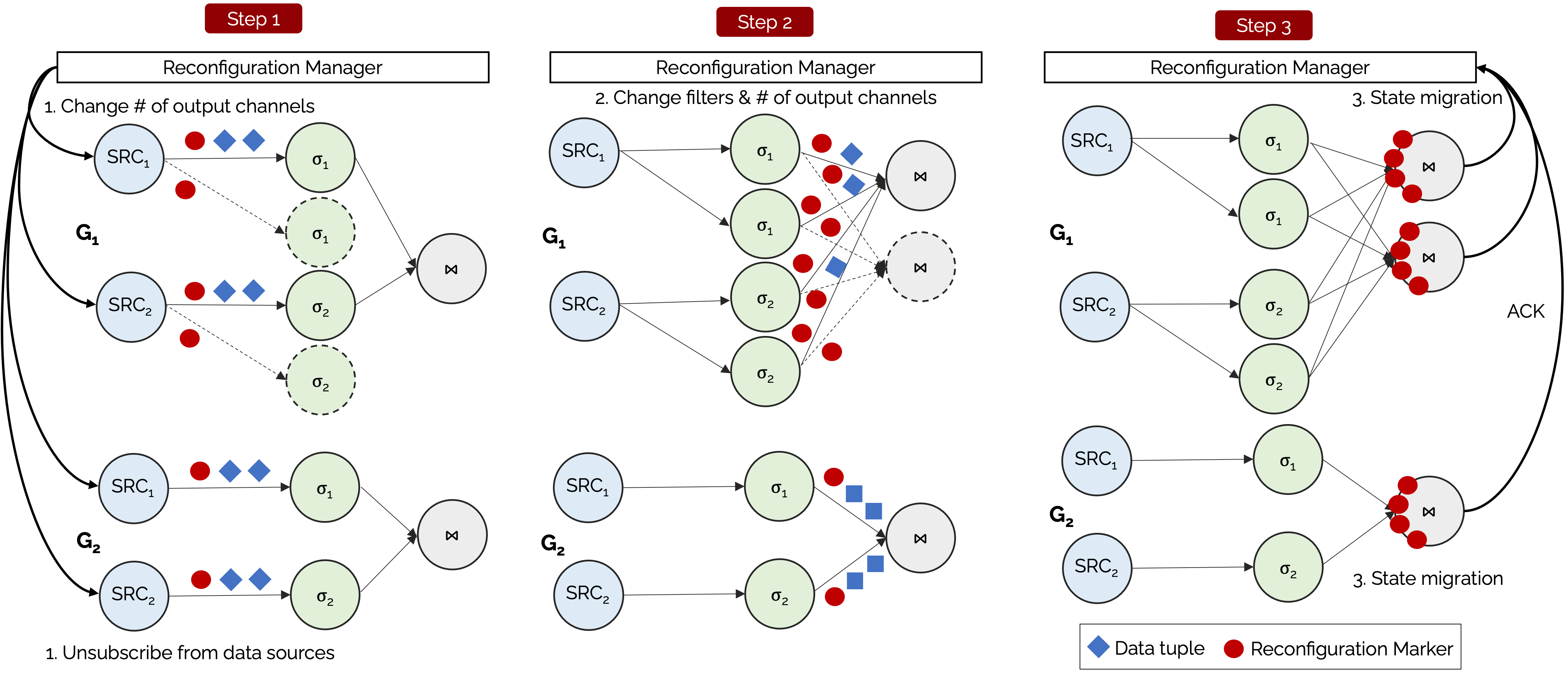}
    \caption{Reconfiguration steps for merging two queries}
    \label{fig:reconfiguration}
    \vspace{-0.2cm}
\end{figure*}

Fig.~\ref{fig:merge-monitoring} illustrates the mechanism for three queries that share a subpipeline consisting of a filtering stage followed by a join operation. The queries in the figure use range filters for one of the input streams; however, the approach generalizes to other filter types.
We use the following cost model for these queries:
$\alpha + \text{selectivity}\cdot(\beta + \gamma \cdot \text{joinMatches})$,
where $\alpha$ (the cost for the source and filter operators), $\beta$ (input cost of join phase), $\gamma$ (output cost of join phase) are fixed parameters.
Hence, the data distribution statistics we must track are the selectivity and the average number of join matches per tuple.

Queries $Q_1$, $Q_2$ and $Q_3$ in Fig.~\ref{fig:merge-monitoring} are executed in isolation.
The \emph{Load Estimator} identifies all non-overlapping filter ranges and selects a responsible group for collecting statistics. 
It then informs the \emph{Reconfiguration Manager}, which adjusts the group’s configuration by (i) enabling the tracking of data distribution statistics within the relevant filter and join tasks and (ii) configuring filter tasks to forward all tuples within the monitored ranges to the join operator. (Section~\ref{sec:reconfiguration} describes the reconfiguration mechanism in detail.)
For a defined duration (specified in event time), filter tasks track selectivity, and join operators measure the average number of join matches produced per tuple (Fig.~\ref{fig:merge-monitoring}(b)). 
For simplicity, the example assumes statistics are collected over the entire processing window, eliminating the need for upsampling for join matches. 
Tasks report the collected statistics to the \emph{Monitoring Service}, which forwards them to the \emph{Optimizer}. 
Then, tasks resume normal operation without further monitoring, and the \emph{Optimizer} uses the gathered statistics in the merging phase (Alg.~\ref{alg:merge}). 
Fig.~\ref{fig:merge-monitoring}(c) also shows how the \emph{Load Estimator} estimates 
the combined load of $Q_1 \cup Q_2$. Assuming the \emph{Optimizer} decided to merge $Q_1$ and $Q_2$ into $G_1$, it also presents the calculation of the combined load of $Q_3 \cup G_1$.

The proposed load estimation mechanism enables accurate estimation of any individual group loads and any merged group loads by processing a small sample of tuples.
Thereby, it allows \ouralgorithm{}'s \emph{Optimizer} to perform any number of group merges in a single iteration and minimizes reconfiguration overheads.
Importantly, the correctness of the group responsible for load estimation is not compromised, as the Data Query Model ensures that each query within the group continues to receive only the relevant output tuples. 

\paragraph{Estimating throughput in isolation} 
Given that the load-estimation cost models express the per-tuple processing time, inverting them provides an estimate of the throughput in isolation, which is critical for determining when to trigger the splitting mechanism. 
Each group containing more than one query continuously monitors a fraction of the data to collect the necessary statistics for applying the cost model to its queries. 
For throughput estimation, we need to track statistics only at the query level, and not for every filter range, as we do for load estimation. 
Tasks periodically send the statistics to the \emph{Throughput Estimator}, which estimates per-query throughput and forwards it to the \emph{Optimizer}, which in turn uses it to determine if any group requires splitting.

\paragraph{Monitoring execution metrics}
The \emph{Execution Monitor} tracks several key metrics of the processing engine: 
(i) idle CPU time per task, which is used in the grouping cost equation (Equation~\ref{eq:grouping_cost}), 
(ii) backpressure statistics, which inform the merging and splitting phases (Algorithms~\ref{alg:merge} and \ref{alg:split}), and 
(iii) group throughput, which is used to assess the necessity for group splitting (Algorithm~\ref{alg:split}).
These metrics are collected via fast control messages~\cite{fries}, which allow the controller to communicate with tasks without being blocked by data messages.
Notably, the \emph{Execution Monitor} does not introduce considerable overheads as most of the metrics of interest are already tracked in SPEs for system monitoring~\cite{flink-metrics, kstream-metrics, storm-metrics}, scheduling~\cite{flink-resource-allocation-strategy, storm-scheduler}, and autoscaling~\cite{heinze_autoscaling_2014, google-dataflow-autoscaling}.

\section{On-the-fly Reconfiguration of Query Groups }
\label{sec:reconfiguration}

\ouralgorithm{} reconfigures query plans on the fly. Concretely, given the query groups and the number of parallel subtasks per group determined by the \emph{Optimizer} and the \emph{Resource Manager}, respectively, the \emph{Reconfiguration Manager} orchestrates the necessary changes to the execution plans. There are 4 types of reconfiguration operations in \ouralgorithm{}: group merging, group splitting, changing a group's parallelism, and enabling statistics monitoring for the load estimation. Following prior work~\cite{fries, DBLP:journals/pvldb/MaiZPXSVCKMKDR18},
we adopt epoch-based reconfiguration. Concretely, we rely on the mechanism proposed by~\cite{fries}, which uses fast control messages to optimize epoch-based reconfiguration while still guaranteeing exactly-once processing.

A data stream is divided into consecutive sets of tuples forming an \emph{epoch}, a time period during which the configuration is the same. 
When a reconfiguration request is issued, the \emph{Reconfiguration Manager} injects a  \emph{marker} into each source subtask. When a subtask receives a marker from an input channel, it performs epoch alignment, waiting until it receives the marker from all input channels. Then, it applies the necessary changes to transition to the new configuration before forwarding the marker downstream. Reconfiguration is complete once the markers reach the last subtasks.

Fig.~\ref{fig:reconfiguration} illustrates the reconfiguration process when merging two groups. Initially, these groups independently read, filter, and join two data streams.
Before merging, each group operates with a parallelism of 1 for all operators. After merging, the new group has parallelism 2.
During step 1, the \emph{Reconfiguration Manager} injects markers into all source subtasks. For source subtasks reading data from the same stream, markers are inserted at the same event time. Source subtasks of group G1 change the number of their output channels from 1 to 2 and forward the marker to both the existing and the new filter subtasks. Meanwhile, sources of Q2 unsubscribe from the data sources and propagate markers downstream.
In the second step, the filter subtasks of G1 adjust their filtering criteria to reflect the union of the previous filters from G1 and G2. Then, they change the number of their output channels and forward the markers. The filter subtasks of G2 simply forward the markers. In step 3, upon receiving markers from all upstream channels, the join subtasks initiate state migration to redistribute state partitions, accommodating the new join subtask. Since filters may have differed between G1 and G2, the join subtask of G2 transfers any state that G1 lacked. Once migration is completed, the join subtasks notify the \emph{Reconfiguration Manager}. While source parallelism remains unchanged here, our technique also supports reconfiguring it.  

Reconfiguration follows a similar process for group splitting and adjusting a group's parallelism. Upon receiving the markers, operators adjust the parallelism level, while stateful operators handle state migration as needed. During group splitting, the source subtasks of the newly formed groups register to the respective data sources. 
Finally, for the lightweight reconfiguration of the load estimation process, subtasks receive the filter ranges that should be monitored and adjust their operation so that they process all tuples in these ranges and not solely the tuples belonging 
to the group's queries.

\section{Evaluation}
\label{sec:eval}

 We experimentally evaluate the scalability and adaptivity of \ouralgorithm{} and its impact on per-query performance.

 \textbf{Implementation details} We implement \ouralgorithm{} on the widely used Apache Flink~\cite{flink-original-paper} stream processing engine (v1.13). The \emph{Optimizer} and \emph{Resource Manager} are integrated into Flink's JobMaster. The \emph{Resource Manager} determines the number of subtasks per group (Definition~\ref{def:resources}), while Flink handles subtask placement across TaskManagers. We additionally implement the on-the-fly reconfiguration mechanism described in Section~\ref{sec:reconfiguration}.
 All baselines are implemented in Flink as well. 

\textbf{Platform} We use 5 two-socket Intel Xeon E5-2660 CPUs (2.20 GHz) and 3 Intel Xeon CPU E5-2640 v2 CPUs (2.00GHz) servers, each with ($8\times2$) threads per socket and $128$ GB of DRAM for the TaskManagers. Additionally, we use 2 two-socket Intel Xeon E5-2680 v4 CPUs (2.40 GHz) servers with ($2\times14$) threads per socket for Flink's JobManager and data generation.

\textbf{Methodology} 
Following prior work on benchmarking streaming systems~\cite{DBLP:conf/icde/KarimovRKSHM18, agnihotri2025pdspbenchbenchmarkingparalleldistributed}, we run data generation on a separate machine to eliminate interference between data production and processing. We collect measurements after the system has warmed up and reached a steady state.

\textbf{Algorithms} We compare \ouralgorithm{} against four baselines:
\begin{enumerate}
    \item \emph{Isolated execution:} Each query runs independently with a separate execution plan. 
    \item \emph{Full-Sharing:} All queries are executed with a single shared query plan.
    \item \emph{Overlap-Sharing:} Queries share execution only when the cost of executing them together is less than executing them separately. This is the optimization algorithm used by AJoin~\cite{ajoin}.
    \item \emph{Selectivity-Sharing:} Queries are grouped based on their selectivity class (high or low selectivity). The selectivity threshold is determined through system micro-benchmarking. This is the optimization algorithm proposed by SWO~\cite{swo}.
\end{enumerate}

For \ouralgorithm{}, the \emph{Optimizer} performs a merging step every minute (see Section~\ref{sec:eval-tuning}). Groups monitor statistics for throughput estimation using a sampling rate of 1\% and report them to the \emph{Monitoring Service} every 10 seconds. We set the duration for monitoring statistics for load estimation such that each task collects statistics for 1000 tuples. Finally, we set the merge threshold (Algorithm~\ref{alg:merge}) as a function of resources. 

\begin{figure*}[t]
\centering
\begin{subfigure}[t]{0.32\linewidth}
\includegraphics[width=1\textwidth, keepaspectratio]{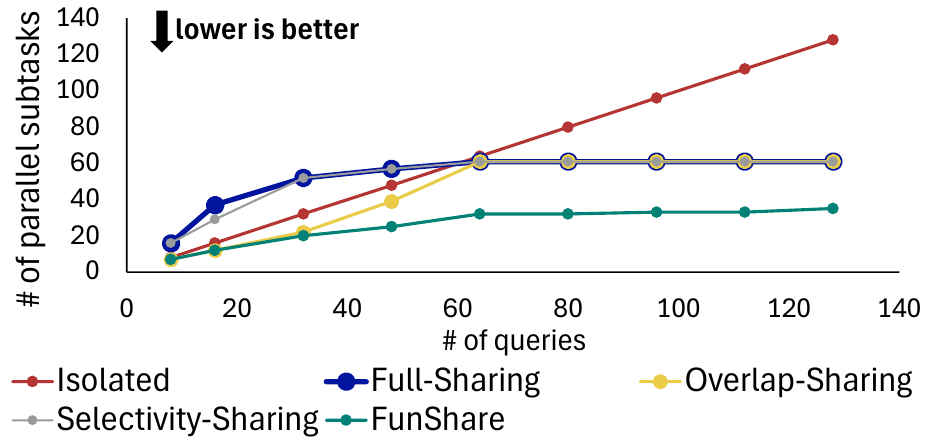}
\caption{W1 - Queries with selectivity 10\%}
\label{fig:exp:resources-joinSameSel10} 
\end{subfigure}
\hfill
\begin{subfigure}[t]{0.32\linewidth}
\includegraphics[width=1\textwidth, keepaspectratio]{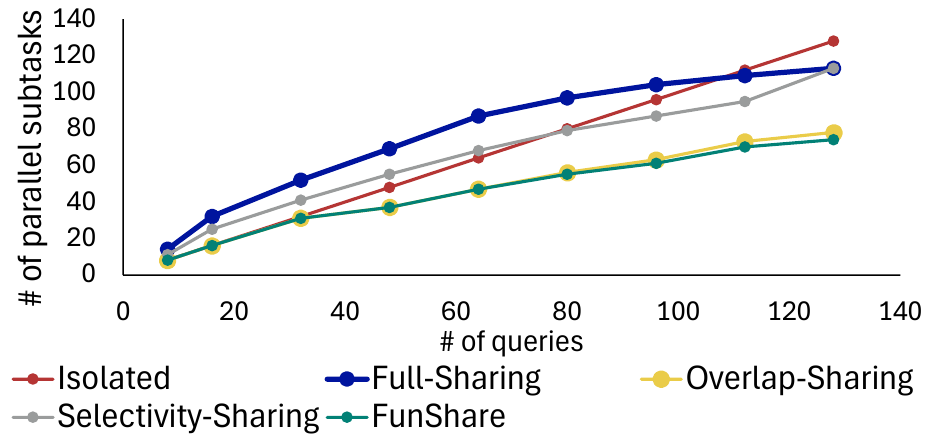}
\caption{W1 - Queries with selectivity 1\%}
\label{fig:exp:resources-joinSameSel1}
\end{subfigure}
\hfill
\begin{subfigure}[t]{0.32\linewidth}
\includegraphics[width=1\textwidth, keepaspectratio]{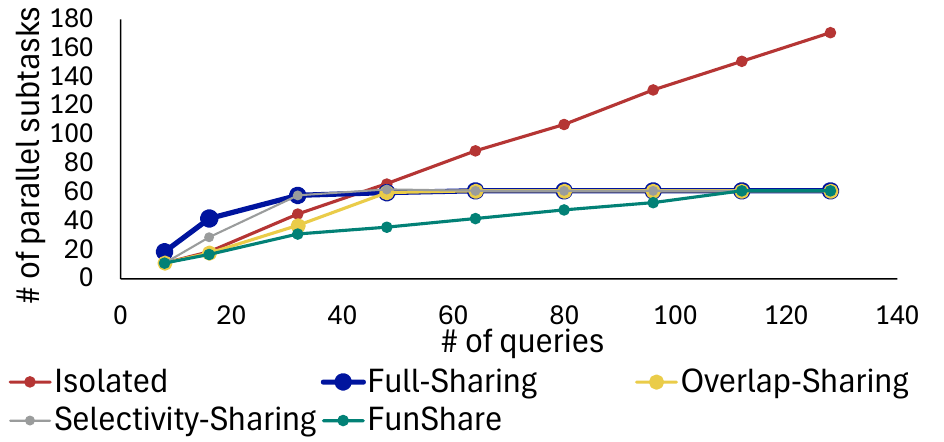}
\caption{W1 - Queries with selectivity 1--20\%}
\label{fig:exp:resources-join-dif-sel}
\end{subfigure}
\centering
\begin{subfigure}[t]{0.32\linewidth}
\includegraphics[width=1\textwidth, keepaspectratio]{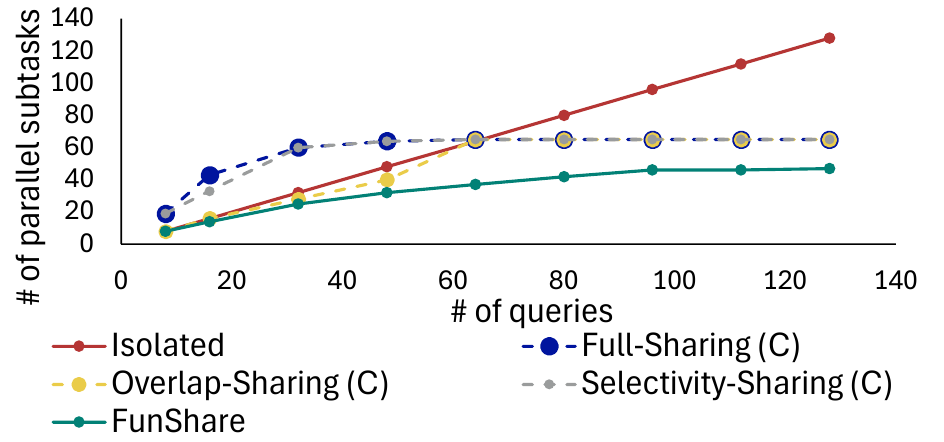}
\caption{W2 - Queries with selectivity 10\%}
\label{fig:exp:resources-downstream}
\end{subfigure}
\centering
\begin{subfigure}[t]{0.32\linewidth}
\includegraphics[width=1\textwidth, keepaspectratio]{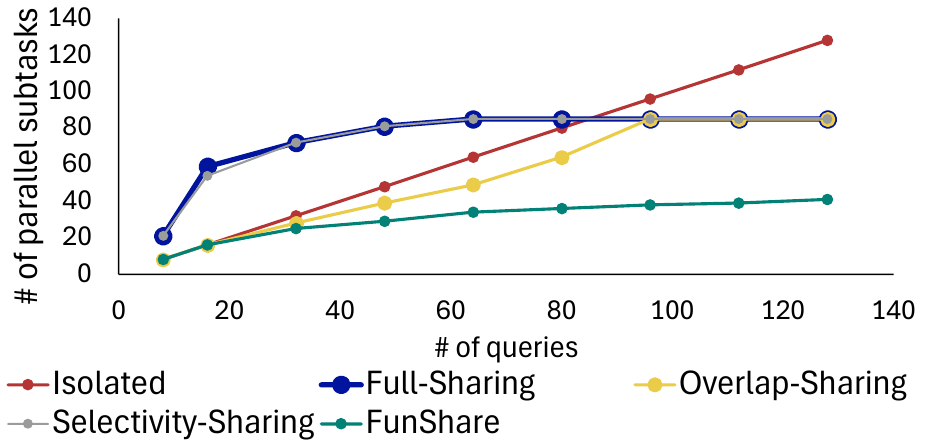}
\caption{W3 - Queries with selectivity 10\%}
\label{fig:exp:resources-vector} 
\end{subfigure}
\caption{\label{fig:exp:resources-scale} Resource usage with fixed throughput. In Fig. (d), the sharing baselines cannot sustain the input rate, so we report a constrained version of them, which never shares the downstream operators}
\vspace{-0.4cm}
\end{figure*}

\textbf{Workload} 
We evaluate three workloads based on the Nexmark Benchmark~\cite{nexmark-beam, Tucker2002NEXMarkA},
which is widely used for evaluating stream processing systems~\cite{flowKV2023, apache-beam-ibm, kalavri-three-steps, Flink-v2-release}.
To control computation overlap between queries, we apply filters with varying selectivities to the base streams across all workloads. We evaluate configurations in which queries have either equal or varying selectivities. Unless stated otherwise, each query, based on its assigned selectivity, selects a random range from the domain of the filter attribute to create variable query overlaps. 
In particular, we evaluate the following workloads and use a window of size 60 and slide 1 sec:\\
\textbf{W1: Windowed join.} Queries in this workload perform a windowed equality join between the Person and Auction streams. 
We use an N-M join to stress shared operators and create varying sharing overlaps\footnote{We add a $Person.favoriteCategory$ field and join it with $Auction.category$ to match users with their preferred auctions.}.
This workload represents a case where all queries have the same structure.  \\
\textbf{W2: Varying downstream tasks.} All queries share a join between the Auction and Bid streams, but differ in downstream operators. We select Nexmark queries 4 ($Q_{CategoryAvg}$: average price per category) and 6 ($Q_{SellerAvg}$: average price per seller)—the only Nexmark queries that share a subplan— and a synthetic query ($Q_{PriceAnomaly}$) with an expensive operation after the join: identifying pairs of auctions with similar descriptions but significantly different prices. We include an equal number of queries from each type.\\
\textbf{W3: Vector similarity.} Queries encode Auction descriptions into vector representations and identify similar auctions. This workload models streaming vector similarity joins, which are an example of a computationally intensive operation and relevant in modern ML workloads.

\textbf{Metrics} 
Our evaluation focuses on the two metrics involved in Problem~\ref{prob:funcisol} (Functional Isolation): resource consumption and throughput. 
Resource usage is measured in terms of the number of parallel tasks, which reliably approximates CPU resource consumption since we maintain a fixed ratio between the number of tasks and physical threads.

\subsection{Resource usage under fixed throughput}
\label{sec:eval-resources}

Fig.~\ref{fig:exp:resources-scale} presents the resources required to sustain the input rate for all workloads and under different selectivity configurations. The input rate is the maximum sustainable when queries run in isolation.
\ouralgorithm{}'s goal is to minimize resource usage and always use fewer resources than Isolated (Problem~\ref{prob:funcisol}).

Fig.~\ref{fig:exp:resources-scale}a, b, and c report the resource requirements for three variations of workload W1. In Fig.~\ref{fig:exp:resources-scale}a and b, all queries have the 10\% and 1\% selectivity, respectively, and in Fig.~\ref{fig:exp:resources-join-dif-sel} selectivities are drawn uniformly from $[1-20\%]$. 
Isolated execution scales linearly with the number of queries. All sharing-based approaches exhibit a plateau in resource usage as they exploit more shared computation. Full-Sharing and Selectivity-Sharing require more resources than isolation at lower concurrency, as they share computation across many queries, creating expensive global plans. If these approaches were constrained to use at most as many resources as isolated execution, they would penalize certain queries. FunShare always uses fewer resources than isolation. It adapts the query groups to the workload and can switch from isolated execution, when necessary to sustain throughput for all queries, to fine-grained groups or full-sharing when resources can be reduced without penalties. The workload of Fig.~\ref{fig:exp:resources-joinSameSel1} exhibits less computation overlap and, hence, \ouralgorithm{} saves fewer resources.

Fig.~\ref{fig:exp:resources-downstream} reports the resource requirements for workload W2 and 10\% selectivity. Full-Sharing, Overlap-Sharing, and Selectivity-Sharing create plans that share the unscalable $Q_{PriceAnomaly}$ and fail to sustain the input rate. Therefore, we report a constrained version of them, labeled (C), which never shares the downstream operators. \ouralgorithm{} discovers the most resource-efficient plan and never shares downstream operators when they slow down the group's shared plan. 

Finally, Fig.~\ref{fig:exp:resources-vector} corresponds to W3, the most compute-intensive workload. At low concurrency, \ouralgorithm{} employs fewer sharing opportunities than in W2 and W3, as the extra work performed by shared plans outweighs the benefits of reducing redundant computation. At higher concurrency, it achieves greater resource savings than in other workloads.

While Overlap-Sharing incurs no penalties here, it may degrade performance when workloads mix queries that can and cannot sustain the input rate (see Section~\ref{sec:eval-adaptivity}).

\textbf{Takeaway: }\ouralgorithm{} needs up to $3.7\times$ less resources compared to the baselines to sustain the input rate and never uses more resources than isolated execution.

\subsection{Throughput under fixed resources}
\label{sec:eval-tput}

\begin{figure}
\centering
    \begin{subfigure}[t]{0.7\linewidth}
    \includegraphics[width=1\textwidth, keepaspectratio]{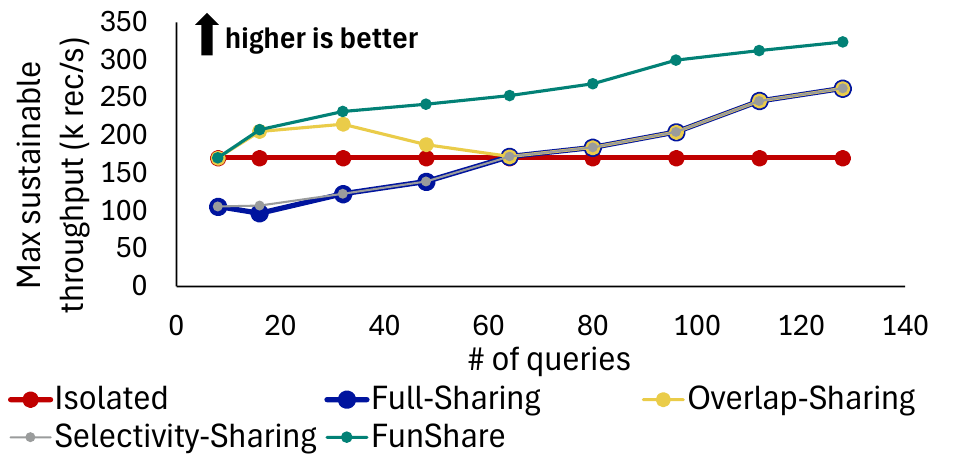}
    \caption{W1 -- queries with selectivity 10\%}
    \label{fig:exp:joinTputSameSel10} 
    \end{subfigure}
    \hfill
    \begin{subfigure}[t]{0.7\linewidth}
    \includegraphics[width=\textwidth, keepaspectratio]{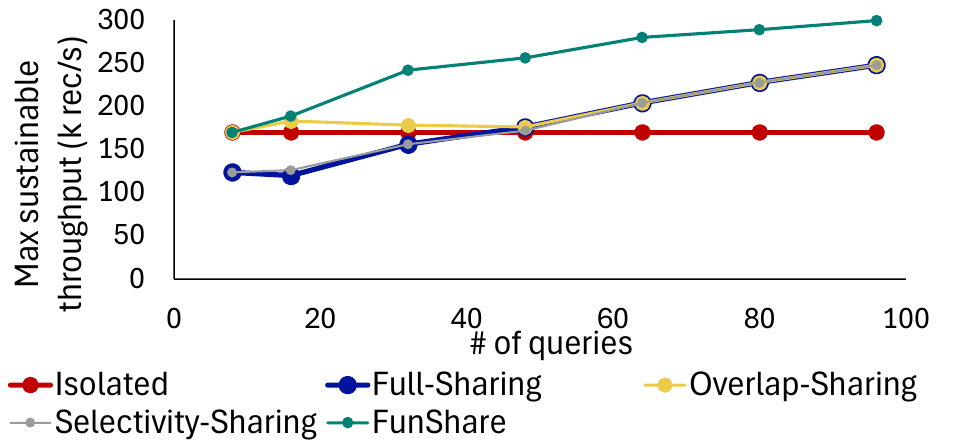}
    \caption{W1 -- queries with selectivity 1--20\%}
    \label{fig:exp:joinTputDifSel} 
    \end{subfigure}

    \caption{Max sustainable throughput under fixed resources}
    \label{fig:exp:scalability-tput}
    \vspace{-0.4cm}
\end{figure}

Fig.~\ref{fig:exp:scalability-tput} reports the maximum sustainable throughput achieved when using as many resources as isolated execution. \ouralgorithm{} aims to achieve at least the same throughput as isolated execution (Problem~\ref{prob:funcisol}). 
We evaluate equal (10\%) and variable (1–20\%) selectivities (Fig. \ref{fig:exp:joinTputSameSel10}, \ref{fig:exp:joinTputDifSel}).
In Fig.~\ref{fig:exp:joinTputDifSel}, we scale up to 96 queries, as higher concurrency requires resources beyond the cluster's capacity for isolated execution.
The throughput for isolated execution is stable. Full-Sharing and Selectivity-Sharing sustain lower throughput than isolated execution for low concurrency; this means they would penalize some queries if the input rate is below 170k records/second and the number of queries below 64 and 48, respectively.
The remaining workloads are omitted for space as they exhibit the same trends: throughput is lower than Isolated when resource costs in Section~\ref{sec:eval-resources} were higher, and when concurrency increases, shared plans leverage more resources and increase throughput.

\textbf{Takeaway:} \ouralgorithm{} never penalizes any query compared to isolated execution and given a fixed amount of resources outperforms all the baselines by up to $1.5-2.1\times$.

\subsection{Adaptivity Experiments}
\label{sec:eval-adaptivity}

\begin{figure}
    \centering
    \includegraphics[width=0.72\linewidth]{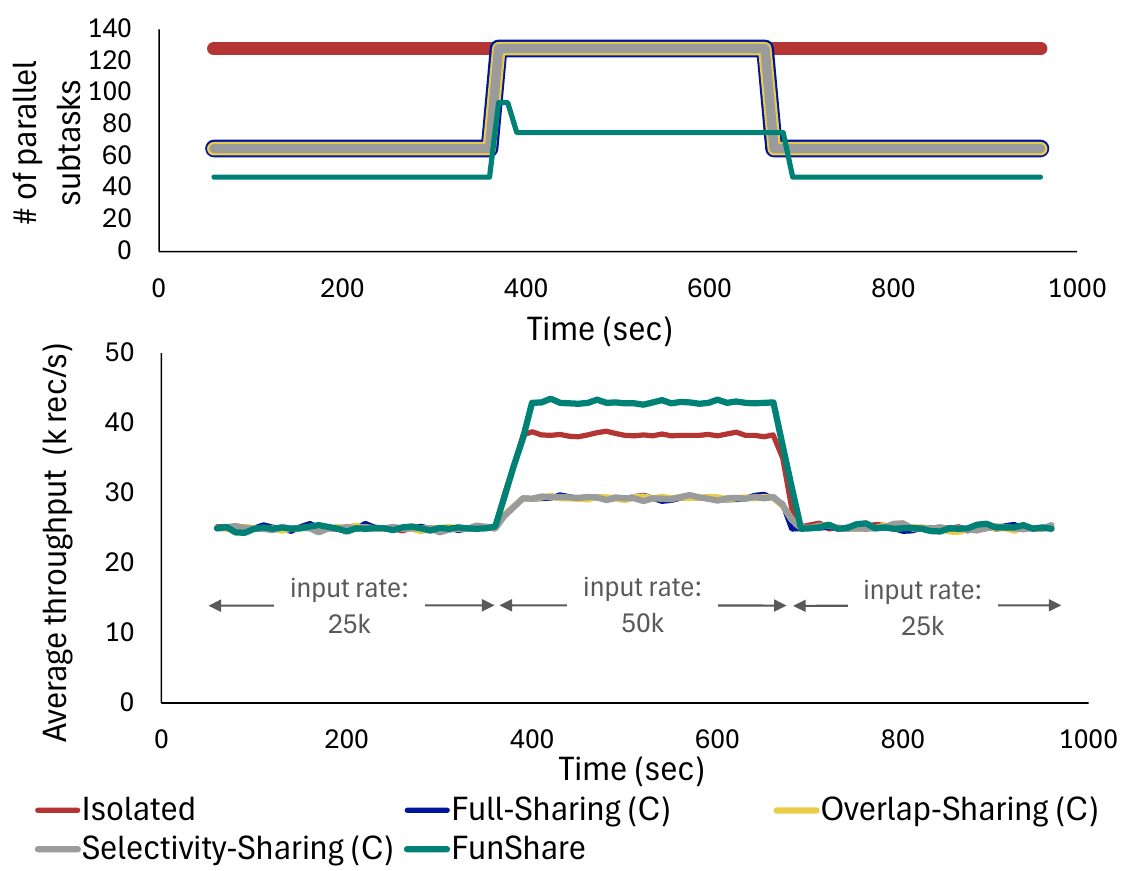}
    \caption{W2 -- Adaptivity to input rate shifts. $Q_{PriceAnomaly}$ queries cannot sustain the increase in input rate}
    \label{fig:exp:adaptivityIR}
\end{figure}

Next, we evaluate the performance of each approach under dynamic workloads by considering two scenarios: variations in the input rate and shifts in data distribution. \ouralgorithm{}'s goal is to minimize resource consumption with the constraint of achieving at least the same throughput as isolated execution.

In Fig.~\ref{fig:exp:adaptivityIR}, we experiment with 128 queries with 10\% selectivity from the workload with different downstream tasks (W2) and varying input rates. This configuration simulates workload shifts where downstream bottlenecks can propagate backpressure to shared operators. Because the sharing baselines cannot sustain the input rate even with maximum resources, we use their constrained versions.
Initially, all approaches sustain the input rate and \ouralgorithm{} has the lowest resource consumption. When the input rate increases, the expensive downstream task of $Q_{PriceAnomaly}$ queries cannot keep up, and hence, these queries cannot sustain the input rate. Sharing baselines penalize queries of type $Q_{CategoryAvg}$ and $Q_{SellerAvg}$, resulting in an average throughput lower than isolated execution. \ouralgorithm{} initially splits the queries of type $Q_{CategoryAvg}$ and $Q_{SellerAvg}$ from their previous group, momentarily increasing resources (at time 360 seconds). During the next merging step, the \emph{Optimizer} adapts sharing groups based on the workload at hand and resource consumption drops (at time 390 seconds). \ouralgorithm{} adapts better than all other approaches to the spike in the workload. Next, as the input rate returns to its initial value, \ouralgorithm{} adjusts the groups again once the merging phase takes place.

In Fig.~\ref{fig:exp:adaptivityDistr}, we experiment with 32 queries from W1 and distribution shifts.
Each query’s filter is a range of random length beginning at the start of the filter attribute’s domain. At 360 seconds, the filter attribute’s distribution changes from uniform to a Zipfian with the most frequent element at the beginning of the domain; at 660 seconds, it shifts to a different Zipfian with the most frequent element in the middle of the attribute's domain. This setup allows us to experiment with three interesting cases: in the first distribution, the amount of computation varies significantly across queries, in the second one, the computation overlap among all queries is very high, and in the third, there are some queries with very high computation and others with very low. Overlap-Sharing employs a single shared plan (like Full-Sharing) because filter ranges are overlapping. 
After transitioning to the first Zipfian distribution, all sharing techniques achieve higher throughput than isolated execution using significantly fewer resources because of the very high computation overlap. When the distribution changes to the second Zipfian, although its average throughput is higher than isolated execution, Selectivity-Sharing penalizes certain queries by up to $1.3\times$.
\ouralgorithm{} dynamically adapts sharing groups, consistently using fewer resources than alternative approaches while maintaining QoS. Specifically, during the first Zipfian distribution, it converges to a full-sharing configuration, whereas under the uniform and second Zipfian, it settles on two distinct, more fine-grained partitionings of queries.

\begin{figure}
    \centering
    \includegraphics[width=0.65\linewidth]{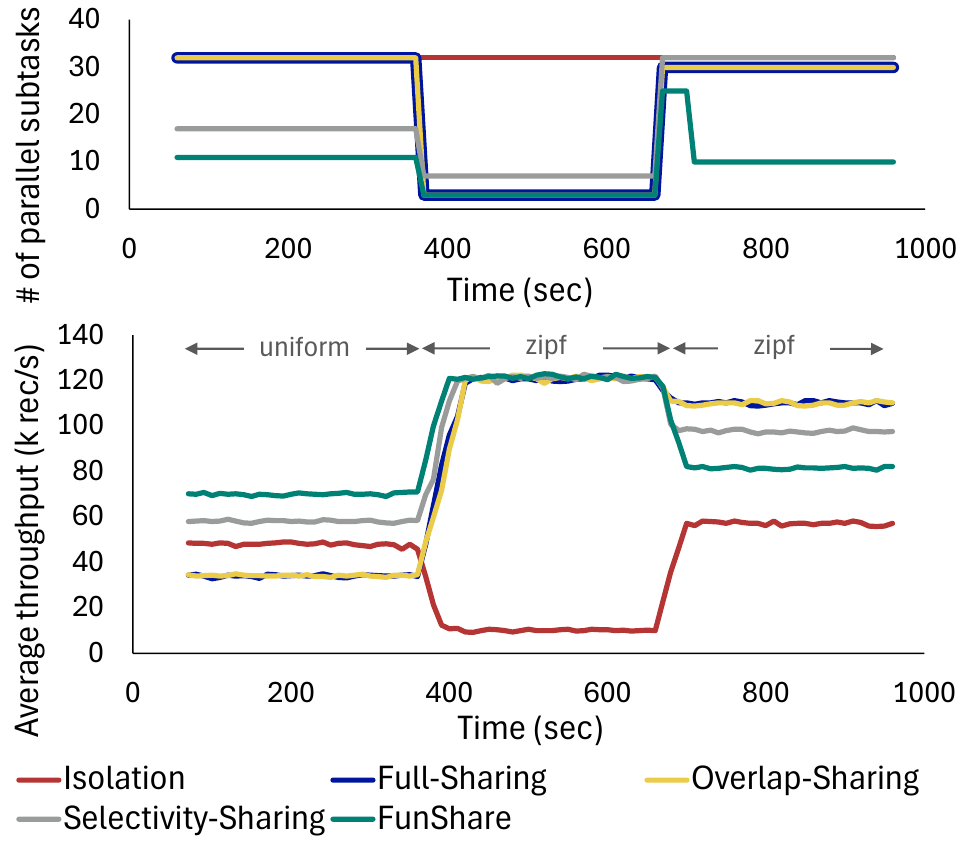}
    \caption{W1 -- Adaptivity to distribution shifts: transitioning from uniform to a Zipfian where all queries select the most frequent key, then to another Zipfian where only some do}
    \label{fig:exp:adaptivityDistr}
\end{figure}

\begin{table}[htbp]
\caption{Reconfiguration delay}
\vspace{-10px}
\label{table:reconfiguration}
\begin{center}
\begin{tabular}{|c|c|c|}
\hline
\textbf{Setup} & \textbf{\textit{1st reconfiguration(s)}} & \textbf{\textit{2nd reconfiguration(s)}} \\
\hline
 Exp. Setup for Fig. \ref{fig:exp:adaptivityIR}  & 1.746 & 1.802  \\ 
 \hline
 Exp. Setup for Fig. \ref{fig:exp:adaptivityDistr} & 1.631 & 1.694\\
 \hline
\end{tabular}
\label{tab1}
\end{center}
\end{table}
\vspace{-10px}

Table~\ref{table:reconfiguration} presents FunShare's reconfiguration delay for Fig. \ref{fig:exp:adaptivityIR} and \ref{fig:exp:adaptivityDistr}. As FunShare continues processing tuples during reconfiguration, this delay is masked and does not impact the reported throughput. The higher delay observed for Fig. \ref{fig:exp:adaptivityIR} is due to the increased number of operators in those query plans.

\textbf{Takeaway:}
\ouralgorithm{} adapts its policy following distribution shifts and can effectively switch between full sharing and fine-grained groups depending on the workload. 

\subsection{Tuning of Merge Phase}
\label{sec:eval-tuning}

\begin{figure}
\centering
\begin{subfigure}{0.57\columnwidth}
\includegraphics[width=1\textwidth, keepaspectratio]{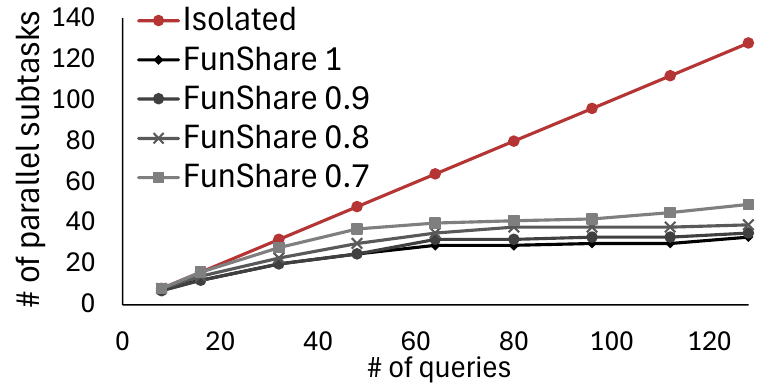}
\caption{}
\label{fig:exp:vary-merge-threshold}
\end{subfigure}
\hfill
\begin{subfigure}{0.41\columnwidth}
\includegraphics[width=1\textwidth, keepaspectratio]{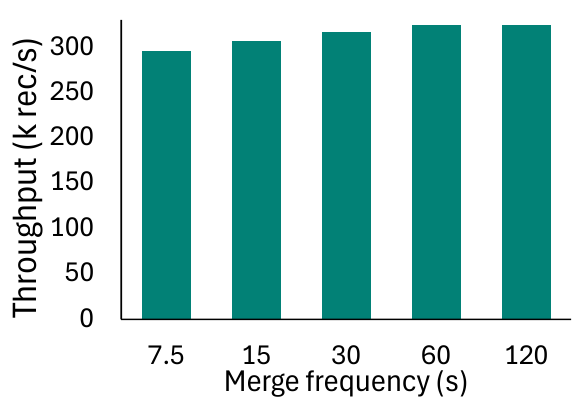}
\caption{}
\label{fig:exp:vary-merge-frequency}
\end{subfigure}
\caption{W1 (selectivity 10\%) -- (a) FunShare's resource consumption with varying merge threshold (b) FunShare's throughput with varying merge frequency (128 queries)}
\end{figure}

Fig.~\ref{fig:exp:vary-merge-threshold} and \ref{fig:exp:vary-merge-frequency} evaluate how the merge threshold and frequency affect \ouralgorithm{}'s performance. We report W1 with 10\% selectivity, but these trends generalize across workloads.
FunShare's resource usage plateaus at 64 queries regardless of the merge threshold, exposing a constantly growing benefit with the number of queries for all cases. Changing the threshold only slightly affects the level of this plateau.
In Fig.~\ref{fig:exp:vary-merge-frequency}, we use 128 queries. \ouralgorithm{} collects runtime statistics at each merge step to detect potential merges, but correctly identifies that no reconfiguration is necessary as the distribution remains stable. Increasing the merge frequency also increases the overhead of monitoring, reducing throughput. We fix the frequency to 60s as it enables quick adaptivity without incurring measurable throughput overhead. Notably, FunShare outperforms all baselines across all frequencies (Fig.~\ref{fig:exp:joinTputSameSel10}).

\textbf{Takeaway:} The merge mechanism exhibits robustness to both its tuning knobs.

\subsection{Latency and Convergence Dynamics of \ouralgorithm{}}
\label{sec:eval-overhead}

Fig.~\ref{fig:exp:latency} compares end-to-end latency for W1 (10\% selectivity). The input rate is set to the max queries sustain in isolation, and techniques are constrained to use as many resources as isolated execution. Full-Sharing and Selectivity-Sharing cannot sustain the input rate for 16 and 32 queries, resulting in unbounded latency. For $\geq 64$ queries, all sharing baselines use a single global plan, which significantly increases latency. \ouralgorithm{} leads to a significantly lower latency increase.

In Fig.~\ref{fig:exp:queue-growth}, we use the same setup as in Fig.~\ref{fig:exp:adaptivityIR} (inp. rate 50 k rec/s) and measure the queue growth rate compared to isolated execution in the presence of backpressure. In isolated execution, the queue growth rate is 35.99k rec/s for $Q_{PriceAnomaly}$ queries and 0 rec/s for the rest of the queries. While sharing baselines reduce the growth rate for $Q_{PriceAnomaly}$, they significantly increase it for $Q_{CategoryAvg}$ and $Q_{SellerAvg}$. In contrast, \ouralgorithm{} reduces it for $Q_{PriceAnomaly}$ without increasing it for any other query. This result implies that the memory and latency growth would also be slower. 

Across all workloads in Section~\ref{sec:eval}, \ouralgorithm{} converged within two merge steps. This reflects the accuracy of our load estimation and the efficiency of the merge step, which together ensure that beneficial merges are rapidly identified and applied.

\textbf{Takeaway:} By partitioning queries into fine-grained groups, \ouralgorithm{} is the only approach that bounds latency increase over isolated execution both in the presence and absence of 
backpressure. Further, it converges quickly thanks to accurate 
load estimation and effective, yet lightweight, optimization.

\begin{figure}
\centering
\begin{subfigure}{0.65\columnwidth}
\includegraphics[width=1\textwidth, keepaspectratio]{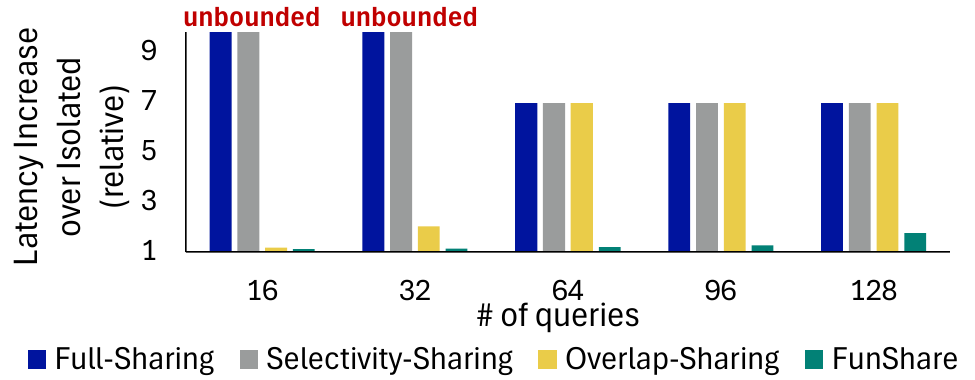}
\caption{}
\label{fig:exp:latency}
\end{subfigure}
\hfill
\begin{subfigure}{0.72\columnwidth}
\includegraphics[width=1\textwidth, keepaspectratio]{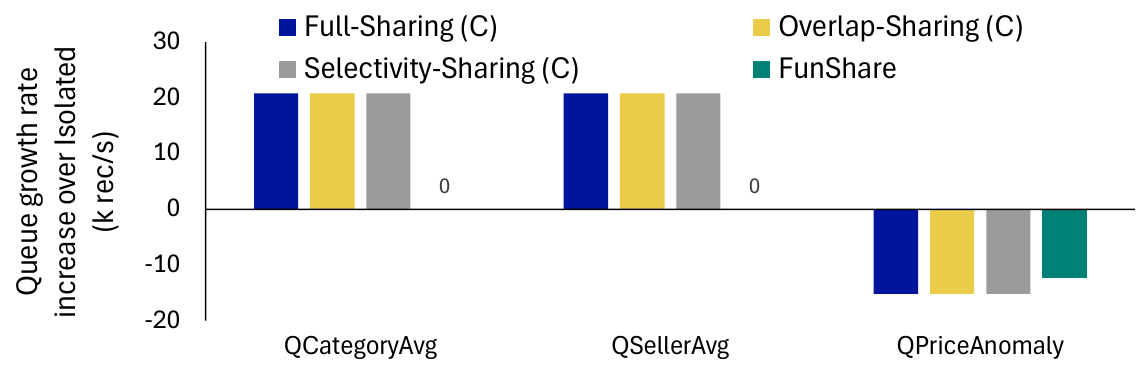}
\caption{}
\label{fig:exp:queue-growth}
\end{subfigure}
\caption{(a) Latency increase compared to isolated execution (b) Queue growth rate increase compared to isolated execution}
\end{figure}

\section{Related Work}

Work sharing is well-studied for analytical~\cite{finkelstein_common_1982, park_using_1988, kalnis_multi-query_2003, giannikis_shareddb_2012, psaroudakis_sharing_2013, roulette} and distributed batch processing \cite{nykiel_mrshare_2010, elghandour_restore_2012, wang_multi_2013}.
Previous techniques in stream processing focus on optimizing aggregate throughput~\cite{chandrasekaran_telegraphcq_2003, astream, ajoin}. SWO~\cite{swo} and AJoin~\cite{ajoin} make sharing decisions based on selectivity and overlap, respectively.
SQPR~\cite{kalyvianaki_sqpr_2011} uses sharing to solve a multi-objective problem aiming to maximize the number of queries the engine executes while minimizing resource usage and balancing load among workers. \cite{on-the-fly-sharing, traub_scotty_2021, DBLP:conf/edbt/SheinCL18} and \cite{hammad_scheduling_2003} optimize shared execution of stream aggregations and joins with varying windows and are orthogonal to our work. All aforementioned approaches do not take into account the individual query requirements.
iShare\cite{iShare} decides to share or not incremental queries based on final work constraints. However, it assumes prior knowledge of the sharing benefit and focuses on incremental processing.
\cite{pham_avoiding_2016} uses work sharing to meet the queries' latency targets, assuming queries are partitioned into priority classes and shed data if they cannot sustain the input rate.

Several prior works study the problem of autoscaling or selecting a query's parallelism~\cite{heinze_autoscaling_2014, dhalion, bilal_2017, kalavri_three_steps_2018, lombardi_elastic_2018, marangozova-martin_multi-level_2019, lian_conttune_2023, pfister_daedalus_2024} and are orthogonal to ours. First, techniques selecting the parallelism can be used to determine the per-query resources that are given as input to our optimizer.
Second, autoscaling techniques can work in synergy with \ouralgorithm{} to adjust the resources allocated to each sharing group when the workload shifts without affecting the effectiveness of groups to meet QoS. In this way, FunShare enhances the effectiveness of autoscaling techniques enabling them to operate on resource-efficient query groups.

\section{Conclusion}
Existing streaming systems achieve performance isolation by employing resource isolation, which results in infrastructure costs that grow with the number of queries. Work sharing is an effective strategy for lowering infrastructure costs by eliminating redundant computations, but penalizes individual queries, compromising Quality-of-Service (QoS). We introduce \ouralgorithm{}, a system that enables resource-efficient execution of streaming queries while preserving individual QoS. \ouralgorithm{} dynamically groups queries based on performance characteristics and continuously adapts sharing decisions to workload changes. 
Our evaluation shows that \ouralgorithm{} reduces resource consumption by $1-10.7\times$ compared to isolated execution while maintaining or improving throughput for all queries.

\section{AI-Generated Content Acknowledgement}
OpenAI’s ChatGPT was used for proofreading. We asked the system to flag potential grammatical errors, typos, and non-idiomatic phrasing in the text and to propose fixes.

\appendix

\subsection{Proofs}
\label{proofs}

\begin{theorem}
Let n be the number of queries running in \ouralgorithm{}, Algorithm~\ref{alg:split} produces a set of sharing groups within at most $n$ executions. Also, in the absence of backpressure, existing sharing groups are not affected.
\end{theorem}

\begin{proof}
We prove this theorem by induction:
\begin{itemize}
    \item n = 1: The initial query group consists of one query, so it is by definition a sharing group.
    \item Suppose that the property holds for $n \leq k$; we prove that it also holds for $n = k + 1$:

    Let $m$ be the number of groups across which the queries are partitioned. Furthermore, let $PQ$ be the penalized queries and $BQ$ the queries experiencing high backpressure. There are three cases:
    \begin{enumerate}
        \item $m = 1$ and $PQ \cup BQ = \emptyset$: then by definition, the group is already a sharing group, and it remains as is, i.e., unaffected by the algorithm.
        \item $m = 1$ and $BQ \cup PQ \neq \emptyset$: if  $BQ \neq \emptyset$, Algorithm~\ref{alg:split} forms singleton groups for queries in $BQ$, else for queries in $PQ$. Either way, let $RQ = BQ, if BQ \neq \emptyset, and PQ, otherwise$. The algorithm forms singleton groups for queries in $RQ$ which are sharing groups. Furthermore, $|g_1 \setminus RQ| \leq k$ so running Algorithm~\ref{alg:split} for at most $k$ more iterations will result in sharing groups for all queries.
        \item $m \geq 2$: Algorithm~\ref{alg:split} runs independently across each group. Each $|g_j| \leq k$, so running the algorithm for each group independently leads to sharing groups within at most $\sum_{j=1}^{m}|g_j| = k + 1$ steps. If any group is already a sharing group and $BQ = \emptyset$, it is not affected.
    \end{enumerate}
\end{itemize}
\end{proof}

\begin{theorem}
Let $m$ be the number of groups in \ouralgorithm{}. Assuming an accurate model for $Load(g)$, linear scalability for operators and a merging threshold that is at most 1, Algorithm~\ref{alg:merge} has the following loop invariant: if the $m$ groups are all sharing groups and queries with backpressure are isolated at the start of the loop, they remain so at the end of the loop as well.
\end{theorem}

\textit{Note that an accurate model for $Load(g)$ and linear scalability are not always possible. For this reason, we set lower, more pessimistic values for the merging threshold, in order to compensate for estimation errors in our model.}

\begin{proof}
In each iteration of Algorithm~\ref{alg:merge}'s loop, since $MERGING\_THRESHOLD \le 1$, the algorithm chooses a pair of groups such that 

$$max(\mathit{GroupingCost}(g_i, g_j), \mathit{GroupingCost}(g_j, g_i)) \le 1$$
We note that, due to the assumption of linear scalability, each group's throughput is
$$T_g = \frac{Resources(g) - IdleResources(g)}{Load(g)}$$

and the maximum sustainable throughput $T^*$ for $g_i \cup g_j$ is

$$ T^* = \frac{Resources(g_i) + Resources(g_j)}{Load(g_i \cup g_j)}$$

We have 
$$\mathit{GroupingCost}(g_i, g_j) \le 1 \implies$$

$$\frac{\frac{Load(g_i \cup g_j) - Load(g_j)}{Load(g_i \cup g_j)}}{\frac{Resources(g_i) + IdleResources(g_j)}{Resources(g_i) + Resources(g_j)}} \le 1 \implies$$

\begin{multline*}
\frac{Load(g_i \cup g_j) - Load(g_j)}{Load(g_i \cup g_j)} \le \\ \frac{Resources(g_i) + IdleResources(g_j)}{Resources(g_i) + Resources(g_j)} \implies
\end{multline*}

\begin{multline*}1 - \frac{Load(g_j)}{Load(g_i \cup g_j)} \le  \frac{Resources(g_i) + IdleResources(g_j)}{Resources(g_i) + Resources(g_j)} \implies
\end{multline*}

\begin{multline*}
\frac{Load(g_j)}{Load(g_i \cup g_j)} \ge  1 - \frac{Resources(g_i) + IdleResources(g_j)}{Resources(g_i) + Resources(g_j)} \implies
\end{multline*}

$$\frac{Load(g_j)}{Load(g_i \cup g_j)} \ge \frac{Resources(g_j) - IdleResources(g_j)}{Resources(g_i) + Resources(g_j)} \implies$$

\begin{multline*}
\frac{Resources(g_i) + Resources(g_j)}{Load(g_i \cup g_j)} \ge \\ \frac{Resources(g_j) - IdleResources(g_j)}{Load(g_j)} \implies
\end{multline*}

$$T^* \ge T_{g_j}$$

Similarly, we can prove that $T^* \geq T_{g_i}$ from $\mathit{GroupingCost}(g_j, g_i) \le 1$

From these equations, we conclude that the new sharing group can sustain a higher throughput, hence maintain functional isolation. Therefore, the resulting group set from the merging step comprises sharing groups.
\end{proof}

The implication of the two theorems is as follows: the first time Algorithm~\ref{alg:split} runs, it produces sharing groups. Then, Algorithm~\ref{alg:merge} performs merging iterations until it reaches a fixed point. From that point on, neither algorithm changes the sharing groups, so \ouralgorithm{} has converged into a grouping. Note that convergence assumes that the data distribution doesn't change during this process; if so, repartitioning the queries is not only necessary but also desirable.

\bibliographystyle{IEEEtran}
\bibliography{bibliography}

\end{document}